%% file: main.tex
\documentclass[
    prd, aps, twocolumn, nofootinbib, showpacs, superscriptaddress
]{revtex4-2}
\usepackage[colorlinks=true,citecolor=green,urlcolor=blue,linkcolor=blue]{hyperref}
\usepackage{graphicx}
\usepackage[utf8]{inputenc}
\usepackage{mathtools}
\usepackage{xcolor}
\usepackage[T1]{fontenc}
\usepackage{graphicx}
\usepackage{latexsym}
\usepackage{rotating}
\usepackage{natbib}
\usepackage{amsmath, amssymb, amsthm, amsfonts}
\usepackage{hyperref}
\usepackage{aas_macros}
\usepackage{soul}
\usepackage{breqn}
\usepackage{xcolor}
\usepackage{tikz}
\usetikzlibrary{decorations.pathmorphing}
\setstcolor{red}

\newcommand{\GB}{\mathcal{R}_{GB}}
\newcommand{\C}{\mathcal{C}}
\newcommand{\limrinf}{\lim_{r \to \infty}}

\newcommand\T{\rule{0pt}{2.5ex}}       
\newcommand\B{\rule[-2.5ex]{0pt}{0pt}}
\newcommand\canuda{\texttt{Canuda}}
\newcommand{\rs}{r_{\text{s}}}

\newtheorem{theorem}{Theorem}
\newtheorem{corollary}{Corollary}[theorem]

\newtheoremstyle{named}{}{}{\itshape}{}{\bfseries}{.}{.5em}{\thmnote{#3 }#1}
\theoremstyle{named}
\newtheorem*{namedtheorem}{Theorem}

\definecolor{orange}{rgb}{1.0, 0.5, 0.0}

\makeatletter
\let\cat@comma@active\@empty
\makeatother

\begin{document}
\title{How Do Spherical Black Holes Grow Monopole Hair?}
\author{Abhishek Hegade K. R.}
\affiliation{Illinois Center for Advanced Studies of the Universe, Department of Physics, University of Illinois at Urbana-Champaign, Urbana, IL 61801, USA}

\author{Elias R. Most}
\affiliation{Princeton Center for Theoretical Science, Princeton University, Princeton, NJ 08544, USA}
\affiliation{Princeton Gravity Initiative, Princeton University, Princeton, NJ 08544, USA}
\affiliation{School of Natural Sciences, Institute for Advanced Study, Princeton, NJ 08540, USA}

\author{Jorge Noronha}
\affiliation{Illinois Center for Advanced Studies of the Universe, Department of Physics,
University of Illinois at Urbana-Champaign, Urbana, IL 61801, USA}

\author{Helvi Witek}
\affiliation{Illinois Center for Advanced Studies of the Universe, Department of Physics, University of Illinois at Urbana-Champaign, Urbana, IL 61801, USA}

\author{Nicol\'as Yunes}
\affiliation{Illinois Center for Advanced Studies of the Universe, Department of Physics, University of Illinois at Urbana-Champaign, Urbana, IL 61801, USA}

\begin{abstract}

Black holes in certain modified gravity theories that contain a scalar field coupled to curvature invariants are known to possess (monopole) scalar hair while non-black-hole spacetimes (like neutron stars) do not. 
Therefore, as a neutron star collapses to a black hole, scalar hair must grow until it settles to the stationary black hole solution with (monopole) hair. 
In this paper, we study this process in detail and show that the growth of scalar hair is tied to the appearance and growth of the event horizon (before an apparent horizon forms), which forces scalar modes that would otherwise (in the future) become divergent to be radiated away. 
We prove this result rigorously in general first for a large class of modified theories, and then we exemplify the results by studying the temporal evolution of the scalar field in scalar Gauss-Bonnet gravity in two backgrounds: (i) a collapsing Oppenheimer-Snyder background, and (ii) a collapsing neutron star background. 
In case (i), we find an exact scalar field solution analytically, while in case (ii) we solve for the temporal evolution of the scalar field numerically, with both cases supporting the conclusion presented above. 
Our results suggest that the emission of a burst of scalar field radiation is a necessary condition for black hole formation in a large class of modified theories of gravity.
\end{abstract}
\maketitle
\section{Introduction}
General Relativity (GR) has so far passed all experimental tests performed to verify its validity in the weak field regime with Solar System observations and the strong-field regime with binary pulsar observations~\cite{2014-Will-Review}.
The detection of gravitational waves by the advanced Laser Interferometer Gravitational-Wave Observatory (LIGO) and Virgo~\cite{2019-GWTC-1,2021-GWTC-2}
has allowed for qualitatively new tests in the highly dynamical, extreme regime of gravity that unfolds in the late inspiral and merger of compact objects.
Testing GR in extreme gravity will naturally involve testing the nature of black holes (BHs) and neutron stars (NS).

BHs are one of the most important predictions of GR, but they are also one of the simplest macroscopic objects in nature. As shown in the celebrated no-hair theorems~\cite{1967-Israel,1968-Israel,1971-Carter,1972-Hawking-No-Hair}, stationary and axi-symmetric BHs in GR can be completely characterised by just three conserved quantities: their mass, charge, and spin.
Several modified gravity theories, however, predict the existence of BHs with
extra ``hair''~\cite{herdeiro2018asymptotically}.
Understanding the nature of BHs and the charges describing BH spacetimes in these modified theories will help us determine how to place constraints on 
them and how to search for deviations from the predictions of GR.

Consider then a general class of modified gravity theories in which a scalar or pseudo-scalar field couples non-minimally  to a curvature invariant.
Modified theories that fall in this class include scalar Gauss-Bonnet theory (sGB)~\cite{1996-Kanti,Mignemi:1992nt}, dynamical Chern-Simons (dCS) gravity~\cite{2003-Jackiw-Pi,2009-Alexander-Yunes} and modified quadratic theories of gravity in general~\cite{Yunes_2013,Cano:2021rey}.

Within this class of theories, several question then naturally arise: 
\begin{itemize}
\setlength\itemsep{0.05cm}
\item[(i)] What types of BHs are allowed?
\item[(ii)] Which of these BHs are the end point of gravitational collapse? 
\item[(iii)] How does the scalar field behave during gravitational collapse? 
\item[(iv)] Is the long-range behavior of the scalar field excited or suppressed during gravitational collapse?
\end{itemize}
These questions are important for various reasons, as their answers determine what BHs look like in a large class of modified theories of gravity, and whether these modifications are observable. 

Some of these questions have been addressed before in the context of specific modified gravity models.
The answer to question (i) has been searched for in various theories, both analytically in the small coupling approximation and numerically in stationary spacetimes.
More specifically, static and slowly rotating BH solutions in the small coupling approximation were found in~\cite{1996-Kanti,2011-Stein-Yunes,Sotiriou-Zhou-2014,Sotiriou:2014PRD,Ayzenberg:2014aka,2011-Pani-Macedo-Cardoso,2015-Maselli} for sGB theory and in~\cite{Yunes:2009hc,Yagi:2012ya,2011-Pani-Macedo-Cardoso,2011-Stein-Yunes,2011-Pani-Macedo-Cardoso} for dCS gravity, while numerical solutions for static and stationary BHs were calculated in~\cite{Sotiriou-Zhou-2014,Sotiriou:2014PRD,2021-Sullivan} for sGB theory and in~\cite{2018-Delsate} for dCS gravity. All of this work has shown that the scalar field at the event horizon (EH) of a stationary and slowly-rotating BH can either be regular or divergent \cite{1996-Kanti,Yunes:2009hc,2011-Stein-Yunes,Sotiriou-Zhou-2014,Kleihaus:2011tg}, with the divergent behavior usually discarded through the imposition of boundary conditions.

Preliminary investigation to answer questions (ii) and (iii) were initiated
in~\cite{Benkel-2016,Benkel_2017}.
These studies simulated the evolution of the sGB scalar in the background of the Oppenheimer-Snyder (OS) collapse. 
Working in the decoupling approximation, i.e., neglecting backreaction onto the spacetime metric, it was found that the scalar field settles to the regular solution.
For studies which analyze the evolution of the scalar field in sGB and Einstein-dilation Gauss-Bonnet (EdGB) \textit{without} working in the decoupling limit, we refer the reader to Refs.~\cite{Ripley-Pretorius-2019,Ripley-Pretorius-2020,Ripley-Pretorius-2020(1),Ripley-East-2021,Kovacs:2020pns,Julie-2020,Julie-2019,julie2022black}.
Some answers to question (iv) were obtained in~\cite{2013-Yagi-Stein-Yunes-Tanaka,2016-Yagi-Stein-Yunes,Prabhu-Stein-2018,Wagle:2018tyk} for sGB theory and dCS gravity, where a classification of the monopole scalar hair was obtained for isolated and stationary NS and BH spacetimes.

Although impressive, all of this previous work was not able to reach general conclusions in a dynamical setting, which is the focus of this paper.
We begin by focusing on question (i), and thus, 

\vspace{-0.2cm}
\begin{center}
\fbox{\parbox{8cm}{
{\bf We develop a complete classification of all static and spherically symmetric BH spacetimes in a wide class of theories, including but not limited to sGB theory and dCS gravity.  }
}}
\end{center}
\vspace{-0.2cm}

\noindent This classification is based on the behaviour of the scalar field at the EH, which can either be finite (Type 1) or divergent (Type 2). 
This two parameter behaviour of the scalar field in BH spacetimes has been observed in a large number of modified theories~\cite{herdeiro2018asymptotically}, including sGB gravity~\cite{2011-Stein-Yunes,Sotiriou-Zhou-2014} and dCS gravity~\cite{Yunes:2009hc}, and it exhausts all possible BH solutions in wide class of theories under consideration.

With this classification at hand, we then establish a few results. First, 

\vspace{-0.2cm}
\begin{center}
\fbox{\parbox{8cm}{
{\bf We \textit{prove} that if the scalar field is regular at the EH, then the BH is non-extremal.
}
}}
\end{center}
\vspace{-0.2cm}

\noindent A non-extremal BH is defined by its surface gravity being non-zero and finite. Our proof does not require an explicit solution for the field equations; rather, it is based on a novel argument that relies on the fact that if the cross sections of the EH of the BH are compact, then the Gaussian curvature of the cross section cannot be zero~\cite{Chrusciel-Reall-Tod-2005}. This result is important because
extremal BHs have been shown to be unstable to perturbations by scalar~\cite{aretakis2013horizon,2013NonLinearArtekis}, electromagnetic and gravitational perturbations~\cite{2012LuciettiReall,2013Murata}.
We emphasize that this result is {\textit{generic}}, i.e., independent of the specific theory of gravity.

\noindent Second, 

\vspace{-0.2cm}
\begin{center}
\fbox{\parbox{8cm}{
{\bf We show that primary hair is sourced only by a divergent scalar field at the EH. 
}
}}
\end{center}
\vspace{-0.2cm}

\noindent In the exterior spacetime, far away from a compact object, the scalar field can be expanded in powers of $1/r$, where $r$ is a suitable distance measure from the compact object. The coefficient of the leading $1/r$ term in the far-field expansion of the scalar field is called the (monopole) ``scalar hair.'' 
The scalar hair is called ``primary'' if knowledge of the intrinsic parameters of the spacetime, such as its mass and spin, is \textit{not} sufficient to determine the value of the scalar hair, and it is called ``secondary''  otherwise. The existence of hair is important because it controls the magnitude of dipole scalar radiation in compact binaries~\cite{2012-Yagi-Stein-Yunes-Tanaka} 
and, consequently, can lead to observable effects in compact binaries~\cite{Shiralilou:2020gah,Shiralilou:2021mfl}
or binary pulsars.
The fact that the scalar hair is secondary for solutions of Type 1 implies that, if these are the end state of gravitational collapse, then observable effects due to dipole radiation will be controlled by the mass and spin of the compact object, which can then be in principle constrained.  

After these analyses, we address question (ii) by studying whether Type 1 or Type 2 solutions are the end states of gravitational collapse. 

\vspace{-0.2cm}
\begin{center}
\fbox{\parbox{8cm}{
{\bf We \textit{prove} that the scalar field settles to a Type 1 solution during spherically-symmetric, gravitational collapse.
}
}}
\end{center}
\vspace{-0.2cm}

\noindent We arrive at this result by perturbatively expanding the field equations in the coupling constant of the theory, neglecting the back reaction of the scalar field onto the metric, and extending the Kay-Wald theorem~\cite{Wald-1979-Theorem,Kay_1987,dafermos2008lectures} to include a curvature source. A similar but fully numerical analysis was carried out in~\cite{Benkel-2016,Benkel_2017} for the particular case of sGB theory, arriving at the same result. Our results are important because they extend these numerical results to a wide class of theories, and they do so with a mathematical and analytical proof. 

Our extension of the Kay-Wald theorem allows us to address question (iii) and (iv) to determine the behavior of the scalar field during gravitational collapse and whether any long range components are excited. 

\vspace{-0.2cm}
\begin{center}
\fbox{\parbox{8cm}{
{\bf We show that the scalar field remains regular during spherically-symmetric gravitational collapse, relaxing to a Type 1 solution by shedding any otherwise divergent behavior through the emission of scalar waves. 
}
}}
\end{center}
\vspace{-0.2cm}

\noindent As before, we work to first order in an expansion about small coupling, and thus, neglect back-reaction of the scalar field onto the spacetime. We then employ our extension of the Kay-Wald theorem, in conjunction with Price's law~\cite{Price-1972,Gundlach:1993tp}, to determine \textit{how} the scalar field and scalar hair relax to a regular configuration (ie.~through the emission of scalar waves). Moreover, we develop a corollary to our extended Kay-Wald theorem to provide
covariant formulae for the scalar hair in the initial and final states of gravitational collapse, proving that scalar hair before and after collapse is secondary.

To close, we consider all of these questions again, but for the specific case of sGB gravity.
In this theory, the scalar hair of a stationary and axisymmetric NS spacetime is zero~\cite{2013-Yagi-Stein-Yunes-Tanaka,2016-Yagi-Stein-Yunes,Wagle:2018tyk}, but it is non-zero for a stationary and axisymmetric BH spacetime. In the latter, the scalar hair is related to the surface gravity and the Euler number of the bifurcation two sphere~\cite{2016-Yagi-Stein-Yunes,Prabhu-Stein-2018}.
Therefore, during gravitational collapse and BH formation, scalar hair must grow and transition dynamically from its NS value to its BH value. 

We study these scalar field dynamics to first order in perturbation theory and in two background spacetime models: an analytical OS collapse background ~\cite{OS-1939} and a full numerical relativity simulation of the collapse of a NS modeled with a perfect fluid stress-energy tensor. Through this study,

\vspace{-0.2cm}
\begin{center}
\fbox{\parbox{8cm}{
{\bf We show that gravitational collapse leads to Type 1 solutions in sGB theory, with secondary hair generated through the radiation of scalar modes that would be otherwise divergent when the EH crosses the stellar surface.  
}
}}
\end{center}
\vspace{-0.2cm}

\noindent We establish this by first finding an \textit{exact} analytical solution for the dynamical evolution of the scalar field in the interior of the star during OS collapse. This solution explicitly shows that the scalar field is bounded as the surface of the star crosses the EH, leading to a Type 1 solution. Moreover, the solution shows that scalar hair in a NS spacetime is zero because of the presence of certain scalar wave modes that would diverge if the stellar surface were inside the event horizon.

These wave modes, however, begin to be radiated away when the EH first forms, and are completely radiated away before the EH crosses the stellar surface, sourcing non-zero scalar hair for the BH spacetime. 
All of these results are then confirmed numerically by simulating the evolution of the sGB scalar field in the dynamical background of a collapsing NS, again to first order in perturbation theory, and we find excellent agreement with our analytical model. These results are important because they suggest that what is responsible for scalar hair in sGB theory is not the change in topology of the background (from a non-punctured NS spacetime to a punctured BH spacetime), but rather the emergence of an EH in the spacetime (in anticipation to the formation of a curvature singularity). 

The remainder of this paper presents all of these results in much more detail, and it is structured as follows.
Section~\ref{sec:BHs-Action} describes the action and the equations of motion for the theories we are studying and presents the classification of BH solutions and their connection to surface gravity.
Section~\ref{sec:Kay-Wald} describes the perturbative approach and proves that the scalar field settles to a regular solution in a dynamical collapse.
Section~\ref{sec:sGB} studies the dynamics of the scalar field in sGB gravity.
Our conclusions and directions for future work are presented in Section~\ref{sec:Conclusions}.
The signature of the metric is $(-,+,+,+)$ and we use geometric units, $G=1=c$ throughout the paper.

\section{Black Hole Solutions in Modified Theories}\label{sec:BHs-Action}
\subsection{Action and Field Equations}
We consider a general class of theories with a scalar or pseudo-scalar field $\Phi$ which couples non-minimally to gravity. The action for the theories we study takes the form,
\begin{align}\label{eq:Action}
     S &= \int d^4x \sqrt{-g} \left( \frac{1}{16\pi} R + \epsilon\, \Phi \,\mathcal{F}(g,\partial g,\partial^2 g,\ldots) \right.\nonumber \\
     &- \left. \frac{1}{2} \left( \nabla_{\mu}\Phi \nabla^{\mu}\Phi  \right)\right) + S_{\text{matter}},
\end{align}
where, $g$ denotes the determinant of the spacetime metric, $R$ denotes the Ricci scalar, $\nabla$ denotes the covariant derivative,   $\epsilon$ is a coupling constant, $\mathcal{F}$ is any arbitrary curvature scalar and $S_{\text{matter}}$ denotes the action for additional minimally coupled matter fields.
We will require that the curvature scalar $\mathcal{F}$ goes to zero faster than $r^{-4}$ in asymptotically flat spacetimes where $r$ is a suitably defined radial coordinate. This assumption guarantees the convergence of certain integrals (see Sec.~\ref{sec:Scalar-Hair-Formula}).

We can map the above action to some specific modified theories. 
For example, it can be mapped to the action of sGB theory by choosing the coupling constant $\epsilon = \alpha_{\text{sGB}}$,
and the curvature scalar $\mathcal{F}$ to be equal to the Gauss-Bonnet invariant, $\mathcal{F} = \mathcal{R}_{\text{GB}} = R^2 - 4R_{\mu\nu}R^{\mu\nu} + R_{\alpha\beta\gamma\delta}R^{\alpha\beta\gamma\delta}$. 
The scalar field $\Phi = \phi_{\text{sGB}}$ now represents the dilaton field.
We can also map the above action to dynamical 
dCS gravity by choosing, $\epsilon = \alpha_{\text{dCS}}/4$ and the curvature scalar $\mathcal{F}$ is equal to the Pontraygin density,
$\mathcal{F} = R_{\beta\alpha\gamma\delta}{}^{*}R^{\alpha\beta\gamma\delta}$. 
The scalar field $\Phi = \theta_{\text{dCS}}$ now represents the axion pseudoscalar.

Our analysis also applies to theories whose action reduces to Eq.~(\ref{eq:Action}) when the coupling constant is small.
An example of such a theory is EdGB theory, whose action is identical to the one above, but with the second term replaced via
\begin{equation}
\epsilon\, \Phi \,\mathcal{F} \to e^{\gamma\Phi} \mathcal{R}_{\text{GB}}\,;
\end{equation}
when the coupling constant is small, $\gamma\Phi \ll 1$, we can expand the above equation linearly in terms of $\gamma$ to find
\begin{align}
\epsilon\, \Phi \,\mathcal{F} \to
\gamma\, \Phi\, \mathcal{R}_{\text{GB}} +  \mathcal{R}_{\text{GB}}\,,
\end{align}
where the last term is just the Gauss-Bonnet invariant, which is topological and therefore does not contribute to the equations of motion. 
Identifying $\gamma$ with $\epsilon$, we see that EdGB theory reduces to the sGB action for small coupling, which is of the form given in Eq.~(\ref{eq:Action}).

The equations of motion for the action in Eq.~(\ref{eq:Action}) take the form
\begin{align}\label{eq:EOM}
    G_{\mu\nu} + 16\pi \,\epsilon\, \mathcal{C}_{{\mu\nu}} &= 8\pi \left(T_{\mu\nu}^{\Phi} +  T_{\mu\nu}^{\text{matter}} \right) \nonumber, \\
    \Box{\Phi} + \epsilon \,\mathcal{F} = 0. 
\end{align}
The stress energy tensor for $\Phi$ is given by
\begin{equation}
    T_{\mu\nu}^{\Phi} =  \nabla_{\mu}\Phi\nabla_{\nu}\Phi - \frac{g_{\mu\nu}}{2}\left( \nabla^{\delta}\Phi\nabla_{\delta}\Phi\right),
\end{equation}
while the tensor $\mathcal{C}_{\mu\nu}$ is given by
\begin{equation}
    \mathcal{C}_{\mu\nu} := \frac{1}{\sqrt{-g}}\frac{\delta}{\delta g^{\mu\nu}}\left( \int d^4x \sqrt{-g}\, \Phi \mathcal{F}(g,\partial g,\partial^2 g,\ldots) \right) .
\end{equation}
The explicit expressions for
the $\C_{\mu\nu}$
tensor for modified quadratic gravity theories such as sGB and dCS can be found in Ref.~\cite{Yunes_2013}.
The matter stress energy tensor is defined by
\begin{equation}
    T_{\mu\nu}^{\text{matter}} = \frac{-2}{\sqrt{-g}}\frac{\delta}{\delta g^{\mu\nu}} S_{\text{matter}}\, .
\end{equation}
\subsection{Black Hole Solutions}\label{sec:Two-Para-Sol}
Let us now classify the local behaviour of the scalar field near the EH in static BH solutions for the theories we described in the previous section and show how this behaviour is related to the surface gravity of the EH.
The theories we are considering reduce to GR when we set the coupling constant to zero (see Sec. \ref{sec:Kay-Wald}). 
More specifically, we assume that the solution spectrum of the theory is a deformation of the GR solution spectrum that scales with $\epsilon^p$ with $p>0$, such that the solutions reduce \textit{smoothly} to GR as $\epsilon \to 0$.
Because of the existence of this continuous limit, we \textit{assume} that BH solutions exist and the causal structure of static and asymptotically flat BH solutions in these theories is similar to that in GR.
That is, for static BH solutions we assume that,
\begin{enumerate}
\setlength\itemsep{0.05cm}
    \item Static BH solutions contain an EH that is generated by a timelike killing vector.
    \item The EH is a stationary null surface.
    \item Cross sections of the EH are compact.
    \item Curvature scalars are regular on and exterior to the EH (ie.~the singularity is ``hidden'' behind the EH).
\end{enumerate}
\begin{figure}
    \includegraphics[width=0.75\columnwidth]{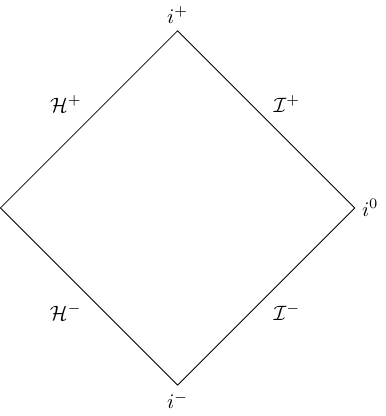}
    \caption{Penrose diagram for static and spherically symmetric BHs for theories whose action is described by Eq.~(\ref{eq:Action}). We assume that the conformal structure of BHs in these theories is similar to the conformal structure of static BHs in general relativity. Note that we only show the part of the conformal diagram which is causally connected to the external universe. In the figure, $\mathcal{H}^{\pm}$ denotes the future and the past event horizons, $\mathcal{I}^{\pm}$ denotes the future and the past null infinity, $i^{\pm}$ denote the future and past time like infinities and $i^0$ denotes the point at spatial infinity.}
    \label{fig:Conformal-Structure-BHs-Alternate}
\end{figure}
We also assume that the scalar field far tends to zero at spatial infinity.
The conformal structure of this spacetime is illustrated in Fig.~\ref{fig:Conformal-Structure-BHs-Alternate}.
A few comments about the above assumptions are in order. 
Assumptions (1) and (2) above can be shown to hold without using field equations~\cite{Carter-1969}.
Assumptions (3) and (4) define BHs as compact objects without naked singularities.

With these assumptions, we can classify the static BH solutions in these theories based on the behaviour of the scalar field on the EH as follows:
\begin{itemize}
\setlength\itemsep{0.05cm}
    \item \textbf{Type 1}. Solutions where the scalar field is regular on the EH.
    \item \textbf{Type 2}. Solutions where the scalar field diverges on the EH.
\end{itemize}
These two types of behaviour of the scalar field occur in a wide class of modified gravity theories.
We refer the reader e.g.~to Table~1 of~\cite{herdeiro2018asymptotically} for a summary of Type 1 and Type 2 solutions in various theories.
This two-parameter behaviour of scalar field solutions was  also found in~\cite{Yunes:2009hc} for slowly rotating BHs in dCS gravity. 
Similar results were also obtained for sGB gravity in~\cite{2011-Stein-Yunes,Sotiriou-Zhou-2014,Sotiriou:2014PRD}. 

Assuming BHs in the class of theories described above exist and respect assumptions (1)--(4), then the metric of a static and spherically symmetric spacetime takes the general form,
\begin{equation}\label{eq:Metric-Spherical-BH}
    ds^2 = -A(\rs) dt^2 + B(\rs) d\rs^2 + \rs^2 d\Omega^2\,.
\end{equation}
We assume that the scalar field respects the symmetries of the metric, $\partial_t \Phi=0$, and the scalar field equation [Eq.~\eqref{eq:EOM}] in these coordinates takes the form,
\begin{equation}\label{eq:Scalar-Field-r}
    \frac{1}{\rs^2 \sqrt{A(\rs)B(\rs)}} \partial_{\rs} \left[\rs^2 \sqrt{\frac{A}{B}} \partial_{\rs} \Phi(\rs) \right] = -\epsilon \; \mathcal{F}.
\end{equation}

The coordinates, $\left(t,\rs,\theta,\phi \right)$ introduced above are not well behaved at the EH. To study the local behaviour of the metric at the EH, it is useful to introduce \textit{Gaussian null coordinates} $(v,r,\theta,\phi)$~\cite{2013-Kunduri,1999-Friedrich}, in which the metric takes the form,
\begin{equation}\label{eq:Gauss-Null}
    ds^2 = -D(r)dv^2 + 2\,dv\,dr + K(r)d\Omega^2\,.
\end{equation}
This metric is regular at the EH and the latter, in the above coordinates, is located at $r=0$. The function $K(r)^{-1}$ denotes the Gaussian curvature of the 2-sphere parametrised by the angular coordinates $\{\theta,\phi \}.$ 
On the EH, $\partial_v$ and $dr$ are null and this condition is equivalent to,
\begin{equation}\label{eq:Local-Met-v}
    D(0) = 0 \implies D(r) = r^{1+p}\,D_0(r), \,p\geq 0 \, ,
\end{equation}
where $D_0(r)$ is a smooth function of $r$ and $D_0(0) \neq 0$.

The surface gravity of the EH can be calculated using the above metric,
\begin{align}\label{eq:Surface-Gravity}
    \kappa &= \Gamma^{v}_{vv} 
    = \lim_{r \to 0}\frac{D'(r)}{2} \,  \nonumber \\
    &= \lim_{r \to 0}\frac{(1+p)r^p D_0(0)}{2},
\end{align}
and, therefore, the EH is degenerate (i.e.,~the surface gravity vanishes) if $p>0$.
Degenerate BH spacetimes are known to be unstable to perturbations, as shown by Aretakis in~\cite{aretakis2013horizon,2013NonLinearArtekis} for perturbations by scalar waves in 4-dimensional BH spacetimes.
This result was then generalised to include electromagnetic and gravitational wave perturbations, as well as higher-dimensional BH spacetimes in~\cite{2012LuciettiReall,2013Murata}.
These results carry through to extremal BHs in the theories that we consider in this paper. 
Moreover, as we show below, regularity of the scalar field at the EH requires that the BH solution should be non-extremal.
But for the sake of completeness, for now, we include these solutions in the discussion below as well. 

The equations of motion \eqref{eq:Action} for the scalar field $\Phi$ in Gaussian null coordinates \eqref{eq:Gauss-Null} take the form,
\begin{equation}
     \frac{1}{K(r)} \partial_r \left[K(r) D(r) \partial_r \Phi(r) \right] = -\epsilon \; \mathcal{F}.
\end{equation}
This equation can be integrated to yield,
\begin{dmath}\label{eq:Scalar-Field-Integrated-v}
     K(r)D(r) \partial_r \Phi(r) - \left.\left[K(r) D(r)\partial_r \Phi(r)\right]\right|_{r = 0} = -\epsilon \int_{0}^{r}\mathcal{F}(x)K(x)\, dx.
\end{dmath}
We now classify the behaviour of the scalar field at the EH and outline how regularity of the scalar field guarantees the non-extremality of the EH.

\subsubsection{Case 1: Scalar field and its derivatives are regular at the event horizon}
\label{sec:Case-1-Scalar-Field}

From Eq.~(\ref{eq:Local-Met-v}), we see that $D(r)$ goes to zero near the EH because $p\geq 0$. Furthermore, recall that $K$ cannot vanish at the EH, since otherwise the Gaussian curvature of the 2-sphere would diverge at the EH.
Therefore, if the scalar field is regular at the EH, then the second term on the left hand side of Eq.~\eqref{eq:Scalar-Field-Integrated-v} must vanish, i.e.,
\begin{equation}
    \epsilon \; Q := \left.\left[K(r) D(r)\partial_r \Phi(r)\right]\right|_{r = 0} = 0.
\end{equation}
We can rewrite Eq.~(\ref{eq:Scalar-Field-Integrated-v}) to read,
\begin{align}\label{eq:Scalar-Field-Dr}
  \partial_r \Phi & =  -\frac{\epsilon} {D_0(r) r^{1+p}K(r)} \int_{0}^{r}\mathcal{F}(x)K(x) dx.
\end{align}
Expanding around the horizon at $r=0$, we can approximate \eqref{eq:Scalar-Field-Dr} as
\begin{equation}
    \partial_r \Phi = \frac{-\epsilon r^{-p} \mathcal{F}(0)}{D_0(0)} + O(r^{1-p}). 
\end{equation}
Therefore, if $\mathcal{F}(0) \neq 0$, then the regularity of the scalar field requires that $p=0$. In turn, this means that the surface gravity is non-zero and the BH must be non-extremal.
On the other hand, suppose $\mathcal{F}(0)$ is equal to zero. In this case we need to look at the gravitational equations \eqref{eq:Action}. 
Studying these equations shows us that the assumption that the scalar field is regular and the fact that $\mathcal{F}(0) = 0 $ leads to a contradiction.
We present a proof of this statement for the case of modified quadratic gravity theories~\cite{Yunes_2013} and $\mathcal{F}(R)$ theories in Appendix~\ref{Appendix:Proof}. 
The proof proceeds by studying the local behaviour of the gravitational field equations and then showing that $\mathcal{F}(0) = 0$ leads to a contradiction with assumptions (3,4) made above.
Therefore, the regularity of the scalar field implies that the BH is non-extremal.

\subsubsection{Case 2: Scalar field or its derivatives diverge at the event horizon}

Let us now go back to Eq.~\eqref{eq:Scalar-Field-Integrated-v}. If the scalar field diverges, then we cannot say whether $Q$ vanishes. Therefore,
\begin{equation}\label{eq:Scalar-Field-Integrated}
    \partial_r \Phi  =   -\frac{\epsilon} {D(r)K(r)} \int_{0}^{r}\mathcal{F}(x)K(x)\, dx + \frac{\epsilon \, Q}{ K(r)D(r)}
\end{equation}
where $Q$ can either be zero or non-zero depending on the rate of divergence of the scalar field at the EH. We note that in this case there is no restriction on the surface gravity of the EH, so it can either be zero or non-zero.

Type 2 solutions with regular BH geometry at the EH are known to exist in a wide class of theories~\cite{herdeiro2018asymptotically}. For example, Bekenstein found an example in the context of Einstein-Maxwell-conformal scalar theory~\cite{1975-Bekenstein}. Although the BH geometry is regular, components of the Ricci tensor and the stress energy tensor projected along null directions cannot remain bounded for this solution. In particular, as argued in Ref.~\cite{1998-Sudarsky-Zannias}, the stress energy tensor for Bekenstein's solution cannot be well-defined at the EH because its projections along null-vectors parallel to the EH diverges.

Does this mean that all Type 2 solutions are unphysical? For the class of theories we are considering, na\"ively projecting the gravitational field equations [Eq.~\eqref{eq:EOM}] along null vectors parallel to the EH horizon might suggest that the components of the Einstein tensor would diverge. But to analyse this carefully, one needs to study the projection of $\mathcal{C}_{\mu\nu}$ along null vectors parallel to the EH. Analysing the properties of this tensor without a specific choice for the curvature coupling $\mathcal{F}$ is not possible. Therefore, we cannot rule out the existence of Type 2 solutions at this stage of the analysis. As we will show in Sec.~\ref{sec:Kay-Wald} however, during a dynamical collapse, the spacetime always settles to solutions of Type 1, at least to first order in perturbation theory.

From the arguments above, we see that the generic local behaviour of $\Phi$ near the horizon can either be regular or divergent.
The reason for this behaviour can be understood by noting that static wave solutions generically diverge near an EH~\cite{Chase-1970}. Given the assumptions we have worked with, BH solutions must fall into one of the two possible behaviours we discussed above.
When the scalar field is regular at the EH, we see that the surface gravity has to be non-zero. 
Moreover, it was shown numerically in Ref.~\cite{Kleihaus:2011tg} that the scalar field diverges for extremal rotating BH solutions in EdGB gravity. This example also hints at the claim that if the BH is extremal, then the scalar field diverges at the EH.
\section{Behaviour of Scalar Hair in Gravitational Collapse}\label{sec:Kay-Wald}
In the dynamical collapse of a NS to a BH, the scalar field will have to settle down to either a Type 1 or Type 2 solution. In the following, we will try to answer this question by considering a perturbative approach to the gravitational field equations \eqref{eq:Action}. More specifically, we adopt the following perturbative approach,
\begin{align}
    g_{\mu\nu} &= g_{\mu\nu}^{(0)} + \epsilon \; g_{\mu\nu}^{(1)} + {\cal{O}}(\epsilon^2)\,,\\
    \Phi &= \Phi^{(0)} + \epsilon \; \Phi^{(1)} + {\cal{O}}(\epsilon^2)\,.
    \label{eq:decomp}
\end{align}
To zeroth order in $\epsilon$, the equations of motion reduce to those of GR plus a minimally coupled massless scalar field $\Phi^{(0)}$.
Stationary BHs in this theory cannot support a non-trivial scalar field if the latter is regular at the EH and the EH geometry is regular~\cite{herdeiro2018asymptotically,Chase-1970,Janis-Newman-Winicour}.
Moreover in dynamical collapse, the scalar field radiates away under suitably physical conditions on the initial data~\cite{1998-Rein-Rendall,Christodoulou:1986zr,1987-Christodoulou,Christodoulou1991TheFO,1994-Christodoulou,1999-Christodoulou,Wald1999}. We therefore set $\Phi^{(0)}$ to zero.

Equations~(\ref{eq:EOM}) and~\eqref{eq:decomp} then imply that the scalar field $\Phi$ scales as $\mathcal{O} (\epsilon)$, and thus $\C_{\mu\nu}$ is also of this order. Because the modification to the metric field equations is proportional to $\epsilon\, \C_{\mu \nu}$ and $T_{\mu \nu}^\Phi \propto (\partial \Phi)^2$, then $\delta g_{\mu \nu}^{(1)}$ must vanish and the first corrections to the metric enter at ${\cal{O}}(\epsilon^2)$. Therefore, to first order in $\epsilon$, the equations of motion become,
\begin{align}\label{eq:EOM-Decoupled}
    G^{(0)}_{ab} - 8\pi \,T_{ab}^{\text{matter}} &=0 \nonumber\,, \\
    \Box{\Phi^{(1)}} + \epsilon \,\mathcal{F}^{(0)} &= 0\,.
\end{align}
From the equations above, we see that we have a scalar field $\Phi^{(1)}$ propagating on a background metric $g^{(0)}_{\mu\nu}$ that satisfies the Einstein equations, to first order in perturbation theory.
To simplify the notation, we will drop the superscripts on the metric and scalar field denoting the order of perturbation henceforth.

Numerical studies carried out in~\cite{Benkel_2017,Benkel-2016} sought to answer which BH configuration the spacetime settles to after dynamical collapse, to first order perturbation theory. 
These studies tracked the profile of the scalar field in sGB theory during OS collapse and in a Schwarzschild spacetime
(see~\cite{Witek:2018dmd} for an analogous computation in Kerr spacetimes.).
Their results suggested that the scalar field eventually settles to the regular Type 1 solution presented above.
A proof of this observed behaviour, however, had not appeared until now. 
In the following sections, we prove that the scalar field settles to a regular Type 1 solution during spherically symmetric gravitational collapse.
We then present covariant formulae for the scalar hair in the initial and end state of gravitational collapse.

\subsection{The BH end state} 

We begin by proving that BHs settle to Type 1 spacetimes after dynamical collapse by developing a generalisation of the Kay-Wald theorem \cite{Kay_1987,Wald-1979-Theorem} and a corollary that uses Price's law~\cite{Price-1972,Gundlach:1993tp}. 
The Kay-Wald theorem~\cite{Kay_1987,Wald-1979-Theorem,dafermos2008lectures} shows that spherically symmetric collapse in GR is stable to linear, scalar field perturbations. 
The Penrose diagram of the collapsing spacetime is shown in Fig.~\ref{fig:Conformal-Spherically-Symmetric-Collapse}, where one observes that curvature scalars only diverge at the central singularity. The precise statement of the theorem is as follows.

\begin{namedtheorem}[Kay-Wald Boundedness]\cite{Kay_1987,Wald-1979-Theorem,dafermos2008lectures}
Let (M,g) be a spherically symmetric spacetime representing the collapse to a spherically symmetric BH with a Penrose diagram as shown in Fig.~\ref{fig:Conformal-Spherically-Symmetric-Collapse}. Let $\Phi$ be a scalar field which is a solution to the Klein-Gordon equation,
\begin{equation}
    \Box \Phi - m^2 \,\Phi = 0, \, m^2\geq0,
\end{equation}
with initial data described on a Cauchy surface prior to the collapse. 
Assume\footnote{For a mathematically precise statement about the decay of the initial data, see Sec. 3 of~\cite{dafermos2008lectures}} that the initial data decays suitably at $i^0$. Then, there exists a constant $C$ such that for all points $p$ exterior to and on the horizon,
\begin{equation}
    |\Phi(p)|<C.
\end{equation}
\end{namedtheorem}

\begin{figure*}[thp!]
\hspace{0.2cm}
\includegraphics[width =0.75\columnwidth]{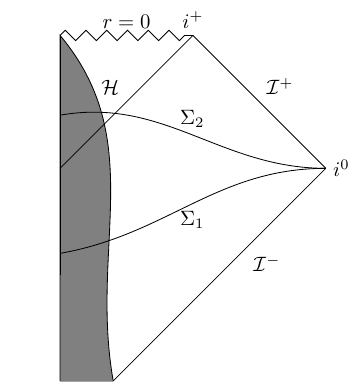}
\hspace{1.5cm}
\includegraphics[width = 1\columnwidth]{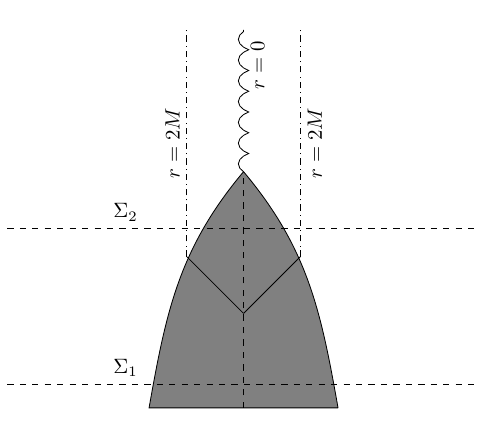}
    \caption{\textbf{Left:} Conformal diagram for a spherically symmetric collapse of matter into a Schwarzschild BH in GR. In this figure, the shaded region is worldtube of the stellar interior, $\Sigma_1$ denotes a Cauchy surface on which we prescribe initial data. $\Sigma_2$ denotes another hypersurface to the future of $\Sigma_1$ used in the proof of Theorem. 1. Theorem 1 in Sec.~\ref{sec:Kay-Wald} shows that solutions to Eq.~(\ref{eq:EOM-Decoupled}) is regular on the region of spacetime to the future of $\Sigma$, on and exterior to the event horizon $\mathcal{H}$.
    \textbf{Right:} Shows a non-compactified version of the figure on the left.}
    \label{fig:Conformal-Spherically-Symmetric-Collapse}
\end{figure*}
With the Kay-Wald theorem stated, we now generalize it to show that to $O(\epsilon)$ the solution to Eq.~(\ref{eq:EOM-Decoupled}) always remains bounded in a spherically symmetric collapse.
\begin{theorem}
Let (M,g) be a spherically symmetric spacetime representing the collapse to a spherically symmetric BH with a Penrose diagram as shown in Fig.~\ref{fig:Conformal-Spherically-Symmetric-Collapse} in GR. Let $\Phi$ be a scalar field which is a solution to Eq.~(\ref{eq:EOM-Decoupled}) with initial data described on a Cauchy surface prior to the collapse. We assume that the initial data satisfies the assumptions of the initial data as stated in the Kay-Wald theorem. Then, there exists a constant $C$ such that for all points $p$ exterior to and on the horizon,
\begin{equation}
    |\Phi(p)|<C.
\end{equation}
\end{theorem}

\begin{proof}
Let us follow the same train of thought as in the proof of the Kay-Wald theorem~\cite{Kay_1987}, starting by proving the theorem in Schwarzschild spacetime. In this spacetime, one can find a particular solution to Eq.~(\ref{eq:EOM-Decoupled}).
Working in Schwarzschild coordinates, the particular solution is given by integrating Eq.~(\ref{eq:Scalar-Field-Integrated}) with $Q=0$,
\begin{equation}\label{eq:Hom-sol-KW}
    \Phi_0(\rs) = \int_{r'=\rs}^{\infty}  \left( \frac{\epsilon}{r'(r'-2M)} \int_{2M}^{r'} \mathcal{F}(x) \, x^2 \, dx \right) dr'.
\end{equation}
The curvature scalar $\mathcal{F}$ is bounded on and outside the horizon and goes to zero far away from the BH.
In particular, we also require that the curvature scalar goes to zero at least as $\mathcal{O}(r^{-4})$. We now rewrite the above equation as,
\begin{align}
    \Phi_0(\rs) = \int_{r'=\rs}^{\infty}  \left( \frac{\epsilon}{r'(r'-2M)} \int_{2M}^{\infty} \mathcal{F}(x) \, x^2 \, dx \right) dr'  \nonumber \\
   - \int_{r'=\rs}^{\infty}  \left( \frac{\epsilon}{r'(r'-2M)} \int_{r'}^{\infty} \mathcal{F}(x) \, x^2 \, dx \right) dr'\,,
\end{align}
the first term on the right-hand side of the above equation goes to zero at least as fast as $\mathcal{O}(r^{-1})$ at spatial infinity and the second term on the right hand side of the equation falls off at least as fast as $\mathcal{O}(r^{-2})$. This means that $\Phi_0$ falls off as $\mathcal{O}(r^{-1})$ at spatial infinity. 
By repeating the local analysis present in Sec.~\ref{sec:Case-1-Scalar-Field}, we see that the solution $\Phi_0$ is bounded on the EH and the bifurcation 2-sphere. Since the curvature scalar is regular throughout the exterior Schwarzschild wedge, $\Phi_0$ is bounded on the exterior Schwarzschild wedge as well.
Therefore, the particular solution $\Phi_0$ is bounded on the exterior Schwarzschild wedge, the EH and on the bifurcation 2-sphere. Let us denote the maximum value of $\Phi_0$ by $C_1$.

We now apply the Kay-Wald theorem to the field $\Tilde{\Phi} = \Phi - \Phi_0$, which satisfies the wave equation, $\Box \Tilde{\Phi} = 0$.
Let $C_2$ be the bound for $\Tilde{\Phi}$ we obtain from the Kay-Wald theorem.
We now use the triangle inequality,
\begin{align}
    |\Phi| \leq |\Phi - \Phi_0| + |\Phi_0| \leq C_1 +C_2
\end{align}
This gives us a uniform bound on $\Phi$ on the Schwarzschild spacetime. 

We now use the causal propagation property of the wave equation (see e.g.~Proposition 7.4.5 in~\cite{hawking_ellis_1973}) in Eq.~(\ref{eq:EOM-Decoupled}) to extend our result to the spherically symmetric collapsing spacetime, $(M,g)$.
To see how this works, suppose that we specify initial data for the scalar field on the hypersurface $\Sigma_1$ as shown in Fig.~\ref{fig:Conformal-Spherically-Symmetric-Collapse}.
Draw another hypersurface $\Sigma_2$ to the future of $\Sigma_1$ as also shown in Fig.~\ref{fig:Conformal-Spherically-Symmetric-Collapse}.
The solution to Eq.~(\ref{eq:EOM-Decoupled}) is well posed~\cite{hawking_ellis_1973,dafermos2008lectures} and in particular propagates at a finite speed in the volume between these two hypersurfaces. By Proposition 7.4.5 of~\cite{hawking_ellis_1973}, the solution to Eq.~(\ref{eq:EOM-Decoupled}) is bounded in this region.
The spacetime to the future of $\Sigma_2$ and outside the EH is just the Schwarzschild spacetime.
Therefore, we can think of the evolution of scalar field to the future of this surface as the evolution of the scalar field in Schwarzschild spacetime, for which we already showed that the solution is bounded.
Therefore, the solution is bounded on and exterior to the EH in $(M,g)$. 
\end{proof}

\begin{corollary}
Let $(M,g)$ be a spherically symmetric spacetime representing the collapse to a spherically symmetric BH with a Penrose diagram as shown in Fig.~\ref{fig:Conformal-Spherically-Symmetric-Collapse} in GR. Let $\Phi $ represent a solution to Eq.~\eqref{eq:EOM-Decoupled}. Then, $\Phi$ settles to a Type 1 regular and spherically symmetric solution to Eq.~\eqref{eq:EOM-Decoupled} after the BH formation.

\end{corollary}
\begin{proof}
By Theorem 1 established above, the scalar field is bounded on and exterior to the EH. 
The spacetime eventually settles to a Schwarzschild BH.
We now use Price's law~\cite{Price-1972,Gundlach:1993tp} which states that the asymptotically flat solutions to the massless Klein-Gordon equation decay away in power law tails in Schwarzschild spacetime; for a mathematically precise derivation of Price's law, see~\cite{2005,angelopoulos2018latetime} and references therein. 
Therefore, the wave modes present in $\Phi$ decay away after a sufficiently long time.
After the wave modes decay away, we are left with a regular static solution of the scalar field equation in Eq.~(\ref{eq:EOM-Decoupled}).
This then proves that, to first order in perturbation theory, the spacetime always settle to a Type 1 regular solution.
\end{proof}
\subsection{Scalar hair evolution}\label{sec:Scalar-Hair-Formula}
We now show how the scalar hair behaves in a collapsing spacetime by developing a second corollary to the above theorem.
Let us begin by defining scalar hair more precisely.
In asymptotically flat spacetimes, we can describe the spacetime region far away from the source by using asymptotically Minkowskian coordinates, $x^{\mu} = (t,x,y,z)$, such that the metric can be expanded as an asymptotic expansion centered around the Minkowski metric $\eta$,
\begin{equation}
    g_{\mu\nu} = \eta_{\mu\nu} \left[ 1 + {\cal{O}}(1/r) \right]\,;
\end{equation}
in this region of spacetime, the components of the curvature tensor scale as $\mathcal{O}(r^{-2})$. 
With this at hand, the equation for the scalar field can be described to first order in $r^{-1}$ by,
\begin{align}
    \Box_{\eta} \Phi &\sim 0\,,
\end{align}
whose solution is simply
\begin{align}
\Phi &= c_0 + \frac{\mu(u,\theta,\phi)}{r} + \mathcal{O}(r^{-2})\,.
\end{align}
In the above equation, $c_0$ is a constant, which we set to zero for asymptotically flat solutions and $u$ denotes retarded time and we also assume that there is no incoming radiation. The scalar hair $\mu$ is defined by the above equation as the $1/r$ coefficient of the far field expansion of the scalar field. The magnitude of the scalar hair quantifies the strength of the scalar field far away from a source and, thus, it controls the amount of dipole radiation emitted by accelerated hairy BHs in quadratic gravity~\cite{2012-Yagi-Stein-Yunes-Tanaka}.

For the BH solutions we considered in Sec.~\ref{sec:Two-Para-Sol}, the explicit formula for the scalar hair can be obtained by asymptotically expanding Eq.~(\ref{eq:Scalar-Field-Integrated}),
\begin{dmath}\label{eq:Scalar-Hair-def-BH}
    \mu = \lim_{r\to \infty} \left. \left(-r^2\,\partial_r \Phi\right)\right. \!\!=\!\lim_{r\to \infty} \left[\frac{{\epsilon}\,r^2}{D(r)K(r)} \int_{0}^{r} \!\! \mathcal{F}(x) K(x)\, dx - \! \frac{\epsilon \,Q \,r^2}{D(r)K(r)} \right].
\end{dmath}
Since the spacetime is asymptotically flat, 
\begin{align}\label{eq:Large-r-1}
    \limrinf D(r) = 1 \,,
    \qquad
    \limrinf \frac{r^2}{K(r)} = 1.
\end{align}
Moreover, in Sec.~\ref{sec:BHs-Action}, we assumed that the curvature scalar decays as $\mathcal{O}(r^{-4})$ and this means that the integral,
\begin{align}\label{eq:Large-r-2}
    \limrinf \int_{0}^{r} \mathcal{F}(x) K(x) dx &= \int_{0}^{\infty} \mathcal{F}(x) K(x) dx ,
\end{align}
is well defined since, 
\begin{align}\label{eq:Large-r-3}
    \limrinf \left[ \mathcal{F}(r) K(r) \right] &\sim \mathcal{O}(r^{-2})\,.
\end{align}
Combining Eqs.~\eqref{eq:Scalar-Hair-def-BH}, \eqref{eq:Large-r-1}, \eqref{eq:Large-r-2} and \eqref{eq:Large-r-3}, we see that,
\begin{equation}\label{eq:Scalar-Hair-BH-Sol}
    \mu = {\epsilon} \left(\int_{0}^{\infty} \mathcal{F}(x)K(x) dx - Q\right).
\end{equation}
Let us now introduce some terminology, which we briefly touched on in the Introduction.
The scalar hair defined in Eq.~\eqref{eq:Scalar-Hair-BH-Sol} is called ``primary' if $Q$ is \textit{not} equal to zero. 
If $Q = 0$, then it is called ``secondary'~\cite{herdeiro2018asymptotically}.
If the scalar hair is primary, then knowing only the intrinsic properties of the compact object, such as its mass, location of the EH and spin, is not sufficient to determine the scalar charge and, consequently, the behaviour of the scalar field far away from the compact object. 
On the other hand, if the scalar hair is secondary, the behaviour of the scalar field far away from the compact object as determined by its scalar charge is completely determined by the intrinsic properties of the compact object. As we have shown in Sec.~\ref{sec:Two-Para-Sol}, for regular BH solutions, $Q=0$. Therefore, Type 1 BH solutions can only possess secondary scalar hair.
\begin{figure*}[thp!]
    \includegraphics[width = 0.75\columnwidth]{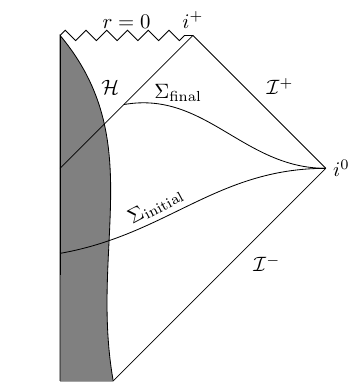}
    \hspace{1.2cm}
    \includegraphics[width = 0.7\columnwidth]{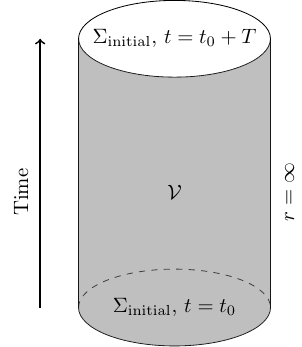}
    \caption{\textbf{Left:} Conformal diagram of a collapsing spacetime. $\Sigma_{\text{initial}}$ and $\Sigma_{\text{final}}$ are spacelike hypersurfaces draw prior to and after the collapse respectively. We assume that spacetime is static around $\Sigma_{\text{initial}}$ and $\Sigma_{\text{final}}$, describing the initially static star which collapses and settles to a final static BH.
    \textbf{Right:} The spacetime volume $\mathcal{V} =[t_0,t_0 + T] \times \Sigma_{\text{initial}}$ is shown in the figure. The boundary $\partial \mathcal{V}$ consists of 3 regions, $\{t_0 + T\}\times\Sigma_{\text{initial}}$, $\{t_0\}\times\Sigma_{\text{initial}}$ and $[t_0,t_0+T]\times \{r = \infty\}$.
    The normals to $\{t_0 + T\}\times\Sigma_{\text{initial}}$ and $\{t_0\}\times\Sigma_{\text{initial}}$ point in opposite directions.}
    \label{fig:Collapse-Hair-Formula}
\end{figure*}

The above expressions can also be generalised to regular spacetimes, such as those of a NS.
Suppose the metric for a spherically symmetric NS spacetime is given by
\begin{equation}\label{eq:Metric-Spherical-NS}
    ds^2 = -A_{\text{NS}}(r) dt^2 + B_{\text{NS}}(r) dr^2 + r^2\, d\Omega^2\,.
\end{equation}
We assume that functions, $A_{\text{NS}}$ and $B_{\text{NS}}$ are regular and non-zero functions at $r=0$.
Using the coordinates introduced above, the equation for the scalar in Eq.~(\ref{eq:EOM}) can now be written as,
\begin{equation}\label{eq:Scalar-EoM-NS}
     \frac{1}{r^2 \sqrt{A_{\text{NS}}(r)B_{\text{NS}}(r)}} \partial_r \left[r^2 \sqrt{\frac{A_{\text{NS}}}{B_{\text{NS}}}} \left(\partial_r \Phi_{\text{NS}}\right) \right] = -\epsilon \; \mathcal{F}_{\text{NS}}.
\end{equation}
We can now integrate the above equation to get,
\begin{dmath}\label{eq:Scalar-Field-NS-1}
    r^2 \sqrt{\frac{A_{\text{NS}}}{B_{\text{NS}}}} \left(\partial_r \Phi_{\text{NS}}\right) - \left.\left[r^2 \sqrt{\frac{A_{\text{NS}}}{B_{\text{NS}}}} \; \left(\partial_r \Phi_{\text{NS}}\right) \right]\right|_{r=0} = -\epsilon \int_{0}^{r} \mathcal{F}_{\text{NS}}(x)\, x^2 \, \sqrt{A_{\text{NS}} B_{\text{NS}}}\, dx
\end{dmath}
We know that the scalar field has to be regular at the center of the star because this is a necessary requirement for regular stellar configurations to exist. 
We further know that the functions $A_{\text{NS}}$ and $B_{\text{NS}}$ are regular at the center of the star.
This means that,
\begin{equation}
    \lim_{r \to 0} \left[r^2 \sqrt{\frac{A_{\text{NS}}}{B_{\text{NS}}}} \; \left(\partial_r \Phi_{\text{NS}}\right) \right]= 0.
\end{equation}
Equation~\eqref{eq:Scalar-Field-NS-1} can now be simplified to get,
\begin{equation}
    \partial_r \Phi_{\text{NS}} = -\frac{\epsilon}{r^2}\sqrt{\frac{B_{\text{NS}}}{A_{\text{NS}}}} \int_{0}^{r} \mathcal{F}_{\text{NS}}(x)\, x^2 \, \sqrt{A_{\text{NS}} B_{\text{NS}}}\, dx\,.
\end{equation}
Far away from the NS we are in an asymptotically flat spacetime so we must have that 
\begin{align}
    \limrinf \left[ A_{\text{NS}}(r)-1 \right] &= 0\,, 
    \\
    \limrinf \left[B_{\text{NS}} (r)- 1 \right] &= 0 \,,
    \\
    \limrinf \left[r^2\,\mathcal{F}_{\text{NS}}(r) \right] &= 0.
\end{align}
We can now repeat the analysis as in Eqs.~\eqref{eq:Large-r-1},~\eqref{eq:Large-r-2} and \eqref{eq:Large-r-3} to find the scalar hair in a NS spacetime
\begin{equation}\label{eq:Scalar-Hair-NS-Sol}
    \mu_{\text{NS}} = {\epsilon} \left(\int_{0}^{\infty} \mathcal{F}_{\text{NS}}(x) \,x^2 \, \sqrt{A_{\text{NS}} B_{\text{NS}}} \,dx\right).
\end{equation}

Let us pause now to compare the scalar hair for the NS spacetime obtained in Eq.~\eqref{eq:Scalar-Hair-NS-Sol} and that obtained for the BH spacetime in Eq.~\eqref{eq:Scalar-Hair-BH-Sol}.
We first note that the coordinates used in these equations are not the same; 
for the BH spacetime we used Gaussian null coordinates [see Eq.~\eqref{eq:Gauss-Null}], while for the NS spacetime we use spherical polar coordinates. If we had used the same coordinates, the expressions would have been functionally identical. However, although the functional form of Eqs.~\eqref{eq:Scalar-Hair-NS-Sol} and \eqref{eq:Scalar-Hair-BH-Sol} is the same, these charges do not evaluate to the same result. This is because the NS spacetime has different metric functions than the BH spacetime, which is why we identified this through a NS subscript. The covariant form of Eqs.~\eqref{eq:Scalar-Hair-NS-Sol} and \eqref{eq:Scalar-Hair-BH-Sol} are presented below, in Eqs.~\eqref{eq:Scalar-Hair-NS-covariant} and \eqref{eq:Scalar-Hair-BH-covariant}.

With the concept of scalar hair defined, and the theorem and corollary established in the previous section, we can now develop a second corollary. 

\begin{corollary}
Let $(M,g)$ be a spherically collapsing spacetime for theories described by Eq.~(\ref{eq:Action}). 
Let $\Sigma_{\text{initial}}$ and $\Sigma_{\text{final}}$ be Cauchy surfaces drawn prior and after the collapse to a BH, as shown in Fig.~\ref{fig:Collapse-Hair-Formula}. We assume that the spacetime near $\{t_0 \} \times \Sigma_{\text{initial}}$ and $\{t_1 \} \times\Sigma_{\text{final}}$ is static.
The spacetime near $\{t_0 \} \times \Sigma_{\text{initial}}$ describes the initially static star that collapses and the spacetime near $\{t_1 \} \times\Sigma_{\text{final}}$ describes the final static BH.
Then, to first order in perturbation theory,
\begin{enumerate}
    \item The scalar hair $\mu_{\text{NS}}$ prior to the collapse is given by,
    \begin{equation}\label{eq:Scalar-Hair-NS-covariant}
        \mu_{\text{NS}} = \frac{\epsilon}{4\pi} \lim_{\epsilon \to 0} \left(\lim_{T\to 0}\frac{1}{T}\int_{[t_0,t_0 + T] \times \Sigma_{\text{initial}}} \mathcal{F}\sqrt{-g}\, d^4x \right).
    \end{equation}
    \item The scalar hair in the final BH spacetime is given by,
    \begin{equation}\label{eq:Scalar-Hair-BH-covariant}
        \mu_{\text{BH}} =  \frac{\epsilon}{4\pi} \lim_{\epsilon \to 0} \left(\lim_{T\to 0}\frac{1}{T}\int_{[t_1,t_1+ T] \times \Sigma_{\text{final}}} \mathcal{F}\sqrt{-g}\, d^4x \right).
    \end{equation}
\end{enumerate}
In particular, scalar hair is secondary before and after gravitational collapse.
\end{corollary}
\begin{proof}
The formula for the scalar hair in the BH spacetime is given in Eq.~\eqref{eq:Scalar-Hair-BH-Sol}. 
By Theorem 1 and Corollary 1.1, we know that the scalar field settles to a regular BH solution after gravitational collapse. This means that $Q = 0$ in Eq.~\eqref{eq:Scalar-Hair-BH-Sol}. Therefore, the scalar hair after gravitational collapse is given by,
\begin{equation}\label{eq:Scalar-Hair-BH-regular}
    \mu_{\text{BH}} = {\epsilon} \int_{0}^{\infty} \mathcal{F}(r)K(r) dr.
\end{equation}
Equations~\eqref{eq:Scalar-Hair-BH-regular} and \eqref{eq:Scalar-Hair-NS-Sol} can be rewritten covariantly as in Eqs.~\eqref{eq:Scalar-Hair-BH-covariant} and \eqref{eq:Scalar-Hair-NS-covariant} respectively.
We now show that Eq.~\eqref{eq:Scalar-Hair-NS-covariant} gives us the correct formula for the scalar hair in the NS spacetime.
We start by writing the scalar field in the NS covariantly as,
\begin{equation}
    \nabla^{\alpha}(\nabla_{\alpha}\,\Phi_{\text{NS}}) + \epsilon \mathcal{F} = 0 \,,
\end{equation}
We now integrate the above equation on a spacetime volume $\mathcal{V} = [t_0,t_0 +T] \times \Sigma_{\text{initial}}$ where $T$ is a small positive number which we use to denote the length of the time interval,
\begin{equation}\label{eq:Stokes-NS-Apply}
   \int_{\mathcal{V}}\left[ \nabla^{\alpha}(\nabla_{\alpha}\,\Phi_{\text{NS}}) + \epsilon \mathcal{F} \right]\sqrt{-g} \,d^4x = 0 \,. 
\end{equation}
We now use Stokes theorem on the first term to get,
\begin{equation}\label{eq:Boundary-NS-Stokes}
    \int_{\mathcal{V}} \nabla^{\alpha}(\nabla_{\alpha}\,\Phi) \,{\sqrt{-g}\,d^4x} = \int_{\partial \mathcal{V}} n^{\alpha} \nabla_{\alpha} \Phi \, \sqrt{h}\, d^3x \, ,
\end{equation}
where, $\partial \mathcal{V}$ denotes the boundary of the volume $\mathcal{V}$, $n^{\alpha}$ denotes the normal vector on the boundary, and $h$ is the determinant of the induced metric on $\partial \mathcal{V}$.
The spacetime of a NS has no interior boundary. Therefore, $\partial\mathcal{V}$ consists of 3 different regions as shown in Fig.~\ref{fig:Collapse-Hair-Formula}.
The integral over the upper and the lower constant time hypersurfaces, $\{t_0 + T\}\times\Sigma_{\text{initial}}$ and $\{t_0\}\times\Sigma_{\text{initial}}$, vanish because we assume that the scalar field is independent of time at this stage of collapse and the normal vectors on these hypersurfaces point in the opposite directions.
Therefore, the only contribution to the boundary integral comes from the boundary at spatial infinity. Furthermore, we consider asymptotically flat spacetime and, thus, the induced metric and the normal at spatial infinity in asymptotically Cartesian coordinates is given by,
\begin{equation}
    ds^2_{i^0} = -dt^2 + r^2 d\Omega^2\,, \, n^{\alpha}_{i^0} = \delta^{\alpha}_r\,.
\end{equation}
We now simplify Eq.~\eqref{eq:Boundary-NS-Stokes} as,
\begin{align}
    \int_{\mathcal{V}} \nabla^{\alpha}(\nabla_{\alpha}\,\Phi) \,\sqrt{-g}\,d^4x &= \int_{t=t_0}^{t_0 + T}\underbrace{\lim_{r\to\infty}\left[ r^2\partial_r \Phi_{\text{NS}}(r)\right]}_{-\mu_{\text{NS}}}  dt\, d\Omega \nonumber ,\\
    &= -4\pi T \mu_{\text{NS}} \, .
\end{align}
Using the above equation to simplify Eq.~\eqref{eq:Stokes-NS-Apply} one finds
\begin{align}
    -4\pi T \mu_{\text{NS}} + \int_{\mathcal{V}}\epsilon \,\mathcal{F}\sqrt{-g} \,d^4x = 0 \,,
\end{align}
which implies that,
\begin{equation}
    \mu_{\text{NS}} = \frac{\epsilon}{4\pi}  \left(\lim_{T\to 0}\frac{1}{T}\int_{[t_0,t_0 + T] \times \Sigma_{\text{initial}}} \mathcal{F}\sqrt{-g}\, d^4x \right),
\end{equation}
where we take the limit of $T \to 0$ because we are in an initially static configuration which is about to collapse and the integral over $t$ factors out only over an infinitesimal interval.

This formula gives the covariant version of NS hair which we have written down in Eq.~\eqref{eq:Scalar-Hair-NS-covariant}.
Note that in Eq.~\eqref{eq:Scalar-Hair-NS-covariant} we take the limit as $\epsilon$ goes to zero because our corollary only holds to first order in perturbation theory.
Repeating the same line of arguments as above, one can also rewrite $\mu_{\text{BH}}$ as given in Eq.~\eqref{eq:Scalar-Hair-BH-covariant}.
\end{proof}
For curvature couplings which are topological, such as the Gauss-Bonnet scalar and the Pontragyin density, the above integral in Eq.~(\ref{eq:Scalar-Hair-BH-covariant}) can be converted to an integral over the EH only using Stokes theorem. 
Results for the scalar hair were obtained in~\cite{Prabhu-Stein-2018,Wagle:2018tyk,2013-Yagi-Stein-Yunes-Tanaka,2016-Yagi-Stein-Yunes} by using Stokes theorem and Noether's theorem for static BH solutions and NS solutions in dCS and sGB.
The above corollary shows that these results are valid even during a dynamical gravitational collapse.

\section{Scalar Gauss Bonnet Theory and Growth of Scalar Hair} \label{sec:sGB}
The theorem and the corollaries presented in the previous section state that the scalar field settles to a Type 1 solution and they also provide formulae for the scalar hair in the initial and final states of gravitational collapse through Eqs.\ (\ref{eq:Scalar-Hair-NS-covariant}) and (\ref{eq:Scalar-Hair-BH-covariant}).
For curvature couplings that are topological, e.g. the Gauss-Bonnet invariant and the Pontraygin density, these scalar hair formulae can be integrated exactly and classified based on the topology of the spacetime.
In this section, we work with sGB theory as an example to illustrate the dynamics of the scalar field during gravitational collapse.
In sGB, the scalar hair takes on a non-zero value~\cite{Prabhu-Stein-2018,2016-Yagi-Stein-Yunes} for Type 1 BH spacetimes and it is zero in static and regular spacetimes such as that of a NS~\cite{2016-Yagi-Stein-Yunes}.
Therefore, the scalar hair \textit{grows} during gravitational collapse to a BH, the opposite to what happens in scalar-tensor theories. 
The scalar field dynamics will be studied to first order in perturbation theory in two background spacetime models: the analytical OS spacetime and a full numerical relativity-generated collapsing NS spacetime. 
In these backgrounds, we find an exact analytical solution and a numerical solution for the scalar field evolution, and use it to study the dynamics of the scalar field.
\subsection{Oppenheimer-Snyder Spacetime}\label{sec:OS-Spacetime}

The OS spacetime is one of the simplest GR  models of collapse within the general Tolman-Bondi class and it describes the spherically-symmetric collapse of matter into a BH \cite{OS-1939}. More precisely, the OS model describes the collapse of an infinitely large 3-ball of pressureless dust due to  its gravitational force, leading to the formation of a BH. 
The stellar interior is modelled as a contracting Friedmann–Lemaître–Robertson–Walker (FLRW) cosmological spacetime with a pressureless dust stress-energy tensor. 
As a consequence of Birkhoff's theorem,
the exterior spacetime is given by the Schwarzschild spacetime.
In this model, there are no surface energy densities and the metric and its normal derivative are continuous~\cite{Poisson:2009pwt,Kanai_2011,Adler_2005} across the stellar surface.

Kanai et al~\cite{Kanai_2011} introduced a global coordinate system to describe the OS spacetime. These global coordinates are obtained by using Painleve-Gullstrand (PG)-type coordinates in the exterior Schwarzschild spacetime, and then gluing these coordinates to the interior spacetime. The metric in the global coordinates, $x^{\mu} = \left(t,r,\theta,\phi\right)$ is given by
\begin{align}\label{eq:PGmetricGlobal}
    ds^2 &= -dt^2 + [dr + v(t,r)\,dt]^2 + r^2 d\Omega^2 \nonumber ,\\
    v(t,r) &= \begin{cases}
    \frac{2r}{3(-t)} & r<R(t)\\
    \sqrt{\frac{2M}{r}} & r \geq R(t)
    \end{cases}.
\end{align}
In the above equation, $R(t)$ denotes the surface of the star, an expression for which (together with other important surfaces) are given in Table \ref{tab:causal-structure}.
In these coordinates, the fluid ball starts at $r=\infty$ at $t=-\infty$ and it collapses to a singularity at $t=0$. 
The causal structure of OS spacetime is shown in Fig.~\ref{fig:Conformal-Diag-OS}.
The coordinate transformation between these global coordinates to Schwarzschild coordinates in the exterior and FLRW coordinates in the interior is reviewed in Appendix \ref{Appendix:Coordinate-Transformations}.

\begin{figure}
    \includegraphics[width = 0.7\columnwidth]{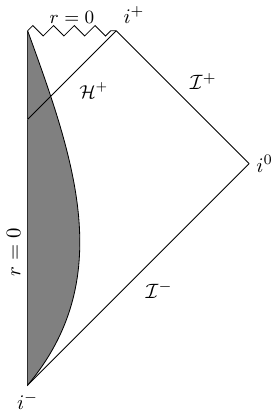}
    \caption{Penrose diagram of Oppenheimer-Snyder spacetime.}
    \label{fig:Conformal-Diag-OS}
\end{figure}
\begin{table*}[thp]
    \centering
    \begin{tabular}{|c|c|c|}
    \hline
    Surface & Global Coordinates & FLRW Coordinates\\
    \hline
    Surface of the Star & $r = R(t) = \left(\frac{9M(-t)^2}{2}\right)^{1/3}$ & $r_c = (\frac{9}{2})^{1/3}M$ \T\B\\
    \hline
      Event Horizon   & $r = H(t) = 3t + 3R(t)$ & $r_c = 3M\left[(\frac{9}{2})^{1/3}-(\frac{-t}{M})^{1/3}\right] $  \T\B\\
    \hline
    Apparent Horizon & $r = -\frac{3t}{2}$ & $r_c = \frac{3M}{2}(\frac{-t}{M})^{1/3} $ \T\B\\
    \hline
    Outgoing Null Geodesics  & $r = 3t - u(-t/M)^{2/3}$ & $r_c = -3\sqrt{a(t)} - u$ \T\B\\
    \hline
    Ingoing Null Geodesics  & $r = -3t + v(-t/M)^{2/3}$ & $r_c = v+ 3\sqrt{a(t)}$ \T\B \\
    \hline
    \end{tabular}
    \caption{Important surfaces in the OS spacetime. The equations for the surfaces are given in both global PG coordinates and in the familiar FLRW coordinates. The coordinate transformation functions between these two coordinate systems in given in Appendix \ref{Appendix:Coordinate-Transformations}. }
    \label{tab:causal-structure}
\end{table*}
We now describe the dynamics of the scalar field in OS collapse.
In sGB theory, the value for the scalar hair in a regular static NS spacetime is zero.
After the NS collapses and settles to a static BH, the scalar hair is non-zero.
To first order in perturbation theory, the exterior spacetime is the Schwarzschild spacetime in both cases.
Why does the hair grow then?
To answer this question, we start by finding the most general spherically symmetric static solution to the sGB scalar field equation in a Schwarzschild spacetime. The sGB equation is given by
\begin{equation}\label{eq:sGB-OS}
     \Box \Phi + \epsilon\, \GB = 0\,,
\end{equation}
and the Gauss-Bonnet scalar in a Schwarzschild spacetime is
\begin{align}\label{eq:GBvalues-Schwarzschild}
    \GB =\frac{48M^2}{r^6}\,.
\end{align}
The scalar equation with a static ansatz in Schwarzschild coordinates becomes
\begin{equation}
    \frac{1}{r^2}\frac{d}{dr} \left[r^2\left(1-\frac{2M}{r}\right)\Phi'(r)\right] + \frac{48\,M^2\,\epsilon}{r^6} = 0.
\end{equation}
The most general solution to the above equation consists of a particular solution and a homogeneous solution, 
\begin{equation}
    \Phi(r) = \epsilon\left(\frac{2}{Mr} + \frac{2}{r^2} + \frac{8M}{3r^3}\right) + \frac{\epsilon\, \alpha}{M^2}{\log \left(1 - \frac{2M}{r} \right)}.
\end{equation}
We can now expand the above solution as $r$ tends to infinity to get the value of the scalar hair,
\begin{equation}
    \Phi(r) \sim \frac{2\,\epsilon(1-\alpha)}{r M } + \mathcal{O}(r^{-2}),
\end{equation}
which means that the value of the scalar hair is given by,
\begin{equation}\label{eq:NS-Hair-sGB}
    \mu = \frac{2 \epsilon(1-\alpha)}{M}.
\end{equation}

For a NS spacetime, topological arguments were derived in Refs.~\cite{2013-Yagi-Stein-Yunes-Tanaka,2016-Yagi-Stein-Yunes} to show that the scalar field must be hairless.
The argument begins by noting that for regular spacetimes such as that of a NS spacetime, there is no interior boundary or ``holes'' in the spacetime.
For the particular case of sGB, the right-hand side of Eq.~\eqref{eq:Scalar-Hair-NS-covariant} vanishes by the generalised Chern-Gauss-Bonnet theorem because the spacetime is connected~\cite{2016-Yagi-Stein-Yunes}.
This means that $\mu_{\text{NS}}$ must be zero.
This argument can also be generalised to NS spacetimes in dCS gravity.
Because $\mu_{\text{NS}}$ is zero, $\alpha=1$ from Eq.~\eqref{eq:NS-Hair-sGB} for a NS spacetime.
On the other hand, for a BH spacetime, regularity at the horizon means that we must set $\alpha = 0$.
This means that,
\begin{align}\label{eq:Sols-Scalar-Field-Static}
    \Phi_{\rm NS} &= \epsilon\left(\frac{2}{Mr} + \frac{2}{r^2} + \frac{8M}{3r^3}\right) + \frac{\epsilon}{M^2}{\log \left(1 - \frac{2M}{r} \right)}  + {\cal{O}}(\epsilon^2)\,, 
    \\
    \Phi_{\rm BH} &= \epsilon\left(\frac{2}{Mr} + \frac{2}{r^2} + \frac{8M}{3r^3}\right) + {\cal{O}}(\epsilon^2)\,.
\end{align}
Therefore, the presence of the homogeneous solution is the difference between the scalar field solution in a NS spacetime versus a BH spacetime.

From the theorems established in Sec.~\ref{sec:Kay-Wald}, the homogeneous mode diverges at the EH, and thus, it must be radiated away during collapse.
To see how this happens, let us study the evolution of the scalar field in the OS spacetime with the hairless solution as initial data.
In the OS spacetime, the star crosses its EH at $t = -4M/3$ (see Fig.~\ref{fig:OS-Collapse}), and after this time we cannot support the hairless initial data because the field diverges at $r=2M$.
Hence, studying the evolution of the hairless scalar field just before the star collapses to a BH will help us answer why and how the scalar hair is radiated away.

In passing, let us note that in the OS spacetime, the curvature scalars exhibit a discontinuity at the surface of the star. This is because the density inside the star is constant and it drops to zero outside the star.
Therefore, the solution for the scalar field will not be continuous across the surface of the star. Nevertheless, one can still study the dynamics of the scalar field inside the star as it collapses and this will give us valuable insights into the dynamics of the scalar field during gravitational collapse.

Let us then prescribe the hairless scalar field as initial data on a null characteristic outgoing surface $u=u_0$, as shown in Fig.~\ref{fig:OS-Collapse} (dashed red line).
For the initial data in the interior of the star, we choose an arbitrary, smooth and bounded profile  $\sigma_2(v)$, which matches smoothly to the exterior hairless data at the surface. 
For the initial data on the other characteristic $v=v_0$ (dotted green line in Fig.~\ref{fig:OS-Collapse}), we choose an arbitrary smooth and bounded profile $\sigma_1$, which matches continuously to the $\sigma_2$ profile at the point of intersection of the two characteristic surfaces.
The initial data is then
\begin{widetext}
\begin{align}\label{eq:Inidata}
    \Phi_{\text{initial}}(u,v_0) &= \sigma_1(u)\\
    \Phi_{\text{initial}}(u_0,v) &= \begin{cases}
    \epsilon\left(\frac{2}{Mr} + \frac{2}{r^2} + \frac{8M}{3r^3}\right) + \frac{\epsilon}{M^2}\log \left(1 - \frac{2M}{r} \right)\,, & r>R(t(u_0,v)) \\
    \sigma_{2}(v)\,, & r<R(t(u_0,v))
    \end{cases}
\end{align}
\end{widetext}
In the above equation, the radial coordinate $r$ is to be read as a function of $u$ and $v$ evaluated at $v = v_{0}$, i.e.~$r(u,v_0)$. 
The transformation relating $(t,r)$ to $(u,v)$ is given in Sec.~\ref{sec:OS-Exact-Sol} and Appendix~\ref{Appendix:Coordinate-Transformations}.
The initial data described above determines the evolution of the scalar field uniquely to the future of the characteristic surface and up to the EH, which is shown as a dash-dotted blue line in Fig.~\ref{fig:OS-Collapse}.
The future of the characteristic lines up to the event horizon is shown as a shaded region in Fig.~\ref{fig:OS-Collapse}.

\begin{figure}[t!]
    \centering
    \includegraphics[width=1.0\columnwidth]{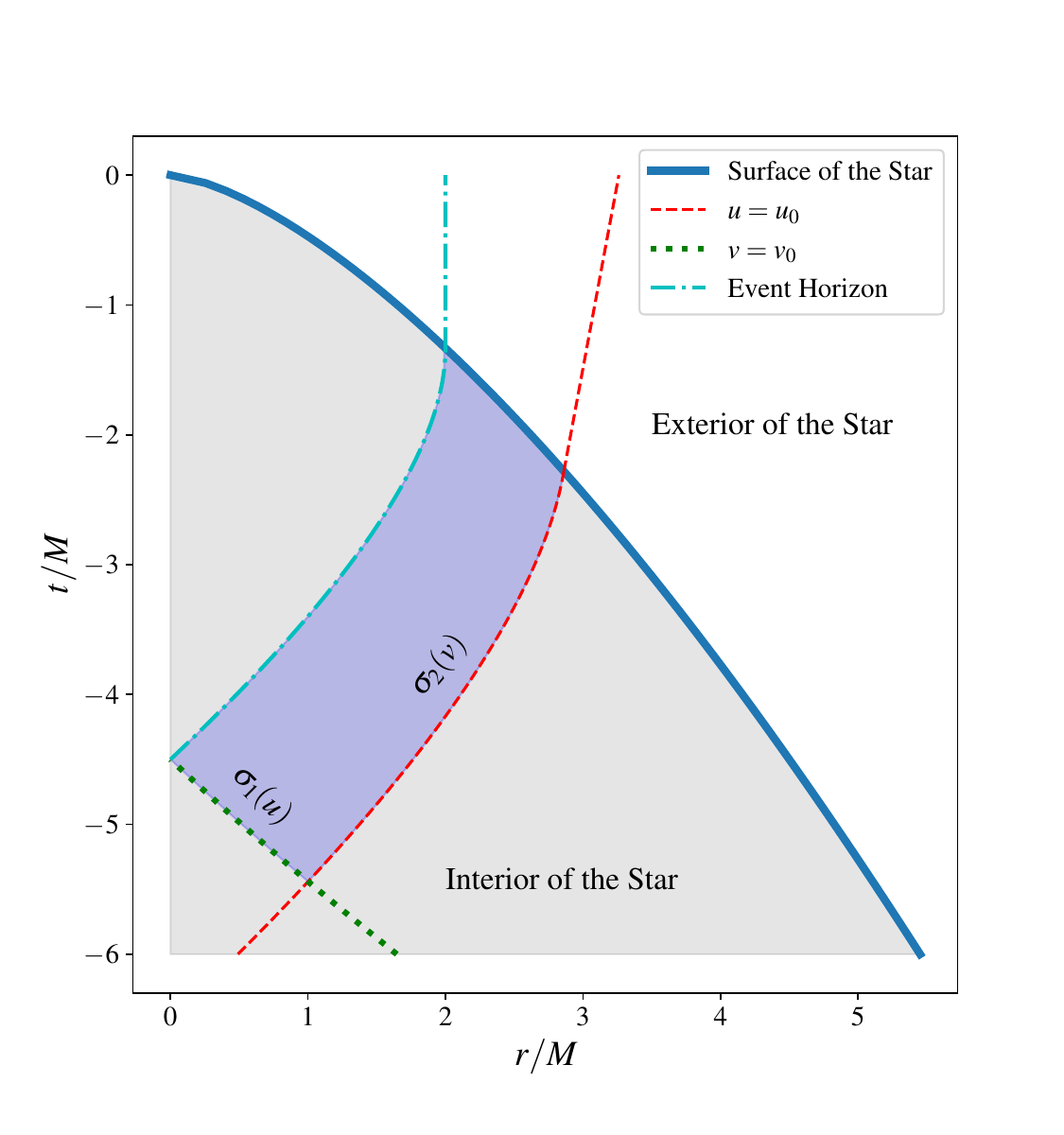}
    \caption{The surface of the star, EH and the null surfaces on which we prescribe initial data is shown in OS collapse. The shaded region in the above figure shows the spacetime region to the future of the null surfaces $u=u_0$ and $v=v_0$ up to the EH and inside the star. The EH forms at $t=-9/2M$ and the surface of the star crosses the EH at $t=-4/3M$, after which the star is causally disconnected from external observers.}
    \label{fig:OS-Collapse}
\end{figure}
We prescribe data on characteristic surfaces because this will enable us to find an exact solution to the scalar field equation that matches the initial data (see Sec.~\ref{sec:OS-Exact-Sol}). 

But first, we need to find the exact solution for Eq.~(\ref{eq:sGB-OS}) inside the star. 
The most general solution is given by the sum of a particular and a homogeneous solution.
To find the particular solution, we start by writing down the Gauss-Bonnet scalar in the OS spacetime, namely
\begin{align}\label{eq:GBvalues}
    \GB &=\begin{cases}
    \frac{48M^2}{r^6} &, \,  r\geq R(t)\,,\\
    -\frac{64}{27 t^4} & , \, r<R(t)\,.
    \end{cases}
\end{align}
The Gauss-Bonnet scalar outside the star is only a function of the radial coordinate, while inside the star it is only a function of the time coordinate.
Therefore, our ansatz for the particular solution will respect this behavior.
Denoting the particular solution inside the star by $\Phi_1^{-}(t)$ and outside the star by $\Phi_1^{+}(r)$, the solution is then
\begin{align}
    \Phi_1^{-}(t) &= -\frac{32 \epsilon }{27 t^2} ,\\
    \Phi_1^{+}(r) &= \epsilon\left(\frac{2}{Mr} + \frac{2}{r^2} + \frac{8M}{3r^3}\right).
\end{align}
The general solution to Eq.~(\ref{eq:sGB-OS}) is then 
\begin{align}
    \Phi^{-}(t,r) &= \Phi_0^{-}(t,r) + \Phi_{1}^{-}(t) ,\nonumber \\
    \Phi^{+}(t,r) &= \Phi_0^{+}(t,r) + \Phi_{1}^{+}(r) ,
\end{align}
where, $\Phi^{\pm}$ is the exterior and interior scalar fields, while $\Phi_0^{\pm}$ satisfies the wave equation
\begin{equation}\label{eq:Wave-Only}
    \Box \Phi^{\pm}_0 = 0.
\end{equation}
To complete the general solution, we therefore must find the homogeneous solution, which we do next.
\subsubsection{Exact Solution for the $\Phi^-$}\label{sec:OS-Exact-Sol}

Let us first focus on the homogeneous solution and then the general solution for the scalar modes inside the star.
We begin by performing a coordinate transformation from PG to double null coordinates, $x^{\mu} = (u,v,\theta,\phi)$. 
The transformation from FLRW coordinates, $x^{\mu}_{\text{FLRW}} = (t,r_c,\theta,\phi)$ to double null coordinates is
\begin{align}
    u &= -3\left(-\frac{t}{M}\right)^{1/3} -r_c \nonumber,\\ 
    v &= -3\left(-\frac{t}{M}\right)^{1/3}+ r_c.
\end{align}
The transformation from FLRW to PG coordinates is given in Eq.~(\ref{eq:PG-FLRW}).
The metric in double null coordinates takes the form
\begin{equation}
    ds^2 = -a^2(t) \, du \, dv + r_c^2\, a^2(t) \,d\Omega^2 .
\end{equation}

With this transformation at hand, we define $\Psi_0^{-} = \Phi_0^- r_c$ so that the wave equation in double null coordinates becomes
\begin{align}
\frac{(u+v)}{2} \partial_{u}\partial_v\Psi_0^{-} + \partial_u\Psi_{0}^- + \partial_v \Psi_{0}^- = 0\,,
\end{align}
whose solution is
\begin{align}
&\implies \Phi_0^- = \frac{1}{r_c}\left[\frac{f'(u)+ g'(v)}{2(u+v)^2} -\frac{f(u)+g(v)}{(u+v)^3} \right].    
\end{align}
This is the most general, spherically-symmetric solution to the wave equation inside the star.
The solution is parameterised by two functions, $f(u)$ and $g(v)$, which characterise the outgoing and the ingoing modes of the scalar field $\Phi_0^-$ respectively.

We now match the above solution with the initial data given in Eq.~(\ref{eq:Inidata}).
To keep mathematical expressions short, we introduce the shorthands
\begin{align}\label{eq:Ini-data-shifted}
    \Tilde{\sigma}_1(u) &= \sigma_1(u) + \frac{32\epsilon}{27\,t^2(u,v_0)} ,\nonumber \\
    \Tilde{\sigma}_2(v) &= \sigma_2(v) + \frac{32\epsilon}{27\,t^2(u,v_0)}, 
\end{align}
where, 
\begin{equation}\label{eq:t-u-v}
    t(u,v) = M\left(\frac{u+v}{6} \right)^3.
\end{equation}
In defining $\Tilde{\sigma}_{1,2}$ we are effectively subtracting the particular solution $\Phi^-_1(t)$ from the initial data in Eq.~(\ref{eq:Inidata}) inside the star.
We can now write down the full solution in the interior of the star,
\begin{widetext}
\begin{align}\label{eq:Sol-Full-Interior}
   \Phi^{-}(u,v) &=\frac{1}{v-u}\left\{\frac{2 J_2 \left(u-u_0\right) \left(u_0+v\right)}{(u+v)^3}+\frac{2 J_1 \left(v-v_0\right) \left(u+v_0\right)}{(u+v)^3}+\frac{\left(v-u_0\right) \left(u_0+v\right){}^2 \Tilde{\sigma}_2(v)}{(u+v)^2}+ \right. \nonumber\\
    & \left. \frac{\left(u_0^2-v_0^2\right) \left[2 u v+v_0 (v-u)+u_0 \left(u-v+2 v_0\right)\right] \Tilde{\sigma _2}\left(v_0\right)}{(u+v)^3}-\frac{\Tilde{\sigma} _1(u) \left(u-v_0\right) \left(u+v_0\right){}^2}{(u+v)^2} \right\} \nonumber \\
    &- \frac{55296\epsilon^2}{M^2(u+v)^6} ,
    \end{align}
\end{widetext}
where,
\begin{align}\label{eq:J-Integrals}
    J_1 &= \int_{u_0}^{u} (v_0 - u') \Tilde{\sigma}_{1}(u') du' \nonumber,\\
    J_2 &= \int_{v_0}^{v} (v' - u_0) \Tilde{\sigma}_2(v') dv'.
\end{align}

Let us now analyze the solution more carefully.
We know that $\sigma_{1,2}$ are bounded functions.
Therefore, from Eq.~(\ref{eq:Ini-data-shifted}) we see that $\Tilde{\sigma}_{1,2}$ are also bounded because $t\neq 0 $.
The full solution in the interior of the star is constructed from $\Tilde{\sigma}_{1,2}$ and their integrals $J_{1,2}$, defined in Eq.~(\ref{eq:J-Integrals}).
Therefore, from Eq.~(\ref{eq:Sol-Full-Interior}), the full solution is also bounded as long as $u+v \neq 0$.
From Eq.~(\ref{eq:t-u-v}), when $u+v$ goes to zero,  we see that $t$ goes to zero, and this is when the curvature singularity forms (see Fig.~\ref{fig:OS-Collapse}).
We then conclude that the solution is bounded to the future of the null surfaces $u=u_0$ and $v=v_0$ up until the EH.
By continuity, the solution is also bounded on the EH.
The future the of null surfaces $u=u_0$ and $v=v_0$ includes the point where the surface of the star crosses its EH and, therefore, the solution is also bounded as the star crosses its EH.

To summarize,
\begin{itemize}
\setlength\itemsep{0.05cm}
    \item The solution in the interior of the star is determined by initial data on characteristic surfaces which are present \textit{before} the EH forms. The spacetime up to this point in the collapse process is regular, and the profiles $\sigma_{1,2}$ are bounded.
    \item Because $\sigma_{1,2}$ are bounded, Eq.~(\ref{eq:Sol-Full-Interior}) tells us that the solution for the scalar field does not diverge to the future of the characteristic surfaces $u = u_0$ and $v=v_0$. In particular, the solution on the surface of the star cannot diverge as the surface of the star crosses its EH.
\end{itemize}
Therefore, the only modes that can escape the star as the star crosses its EH are those that formed \emph{before} the EH forms.
The boundedness of these modes prevent the scalar field from diverging.

\subsubsection{Solution for $\Phi^+$}

Let us now focus on the homogeneous solution and then the general solution for the scalar modes outside the star, so that we can study how the regular modes propagate from the center of the star out to external observers.
The field $\Phi_0^+$ obeys the wave equation in the Schwarzschild spacetime, so defining $\Psi_0^+ = r\Phi_0^+$ and working in Schwarzschild null coordinates, $x^{\mu}_{\text{out}} = \left(u_s,v_s,\theta,\phi \right)$ (different from the null coordinates introduced in the interior), the wave equation becomes
\begin{equation}
    \partial_{u_s}\partial_{v_s} \Psi^+_0 + \frac{M}{2r^3}\left(1-\frac{2M}{r}\right)\Psi^+_0 = 0\,,
\end{equation}
where the effective potential 
\begin{equation}\label{eq:Effective-Potential}
    V(r) = \frac{M}{2r^3}\left(1-\frac{2M}{r}\right),
\end{equation}
is strongly peaked at $r = 8M/3$. 
Therefore, close to the spacetime event when the star goes inside its EH, the above equation can be approximated very well by~\cite{Price-1972,Gundlach:1993tp},
\begin{align}\label{eq:Exterior-Sol-hom}
    \partial_{u_s}\partial_{v_{s}} \Psi^+_0 = 0\,.
\end{align}
The solution to the above equation consists of an ingoing mode and an outgoing mode.

We can now understand how the bounded scalar modes from the interior propagate out and settle to a regular and stationary scalar field configuration.
The solution in the exterior region is given by,
\begin{equation}
    \Phi^{+}(t,r) = \frac{\Psi^{+}_0}{r} + \epsilon\left(\frac{2}{Mr} + \frac{2}{r^2} + \frac{8M}{3r^3}\right).
\end{equation}
The wave modes in $\Psi^{+}_0$ propagate from the surface of the star and reach the peak of the effective potential given in Eq.~\eqref{eq:Effective-Potential}.
Modes with sufficient energy penetrate the potential barrier and propagate out to $\mathcal{I}^+$, the lower energy modes back scatter from the potential.
The energy in the modes that propagate to $\mathcal{I}^+$ decays away as a power law by Price's law.
The modes that back scatter and propagate to $\mathcal{H}^+$ also decay away as a power law~\cite{Gundlach:1993tp}.
Hence, $\Psi^{+}_0$ decays to zero after sufficiently long time. 
This means that,
\begin{equation}\label{eq:OS_final_sol}
    \Phi_{\text{final}} = \epsilon\left(\frac{2}{Mr} + \frac{2}{r^2} + \frac{8M}{3r^3}\right)\,,
\end{equation}
which is a hairy solution.
 
\subsection{Collapsing Neutron Star Spacetime}
\begin{figure*}
    \centering
    \includegraphics[width=0.45\textwidth]{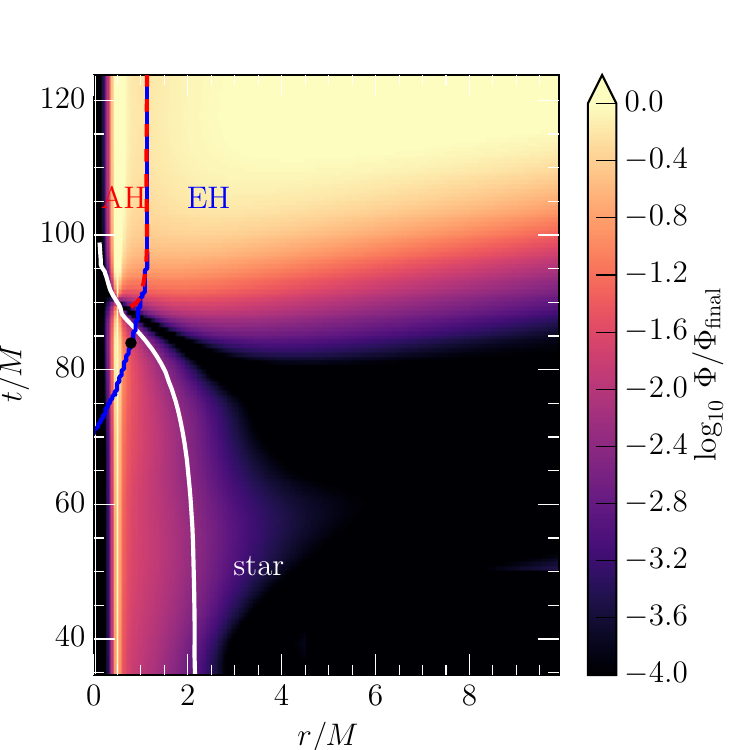}
\includegraphics[width=0.45\textwidth]{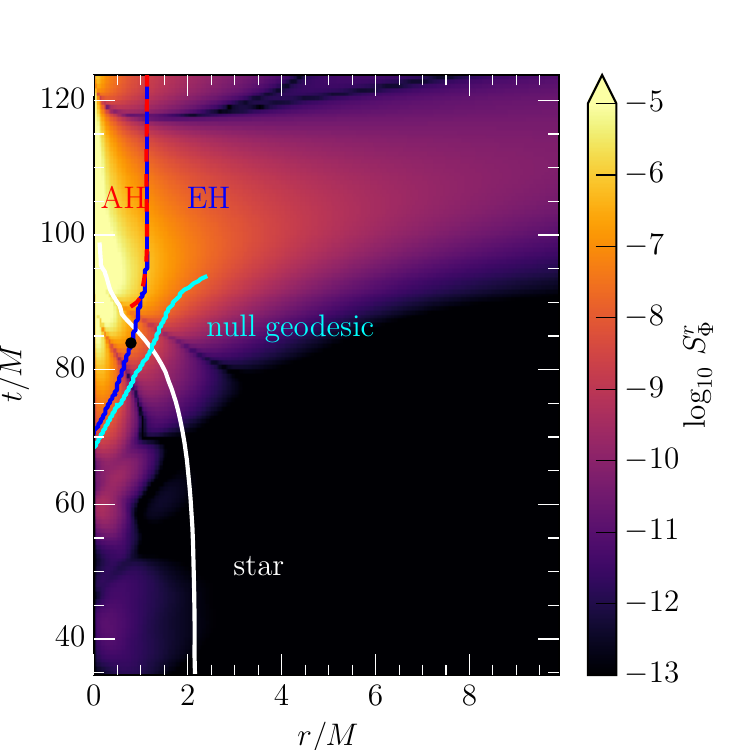}
    \caption{Simulated space-time diagram of a collapsing NS in perturbative sGB gravity. Shown are the surface (white), the event horizon (EH, blue) and the apparent horizon (AH, red). The light blue line refers to null geodesic, and the black dot to the time of peak flux. The time $t$ and radii $r$ shown both refer to individual coordinates. The panels show, {\it (Left) } the value of the scalar field $\Phi$ relative to its final configuration $\Phi_{\rm final}$ [Eq.~\eqref{eq:OS_final_sol}] and {\it (Right) } the scalar energy flux $S^r_\Phi$.
    }
    \label{fig:numerical_collapse}
\end{figure*}
Let us now compare the analytical results obtained in the previous subsection to the behavior of the scalar field in a fully dynamical, and fully relativistic, collapsing NS background. More specifically, we solve the decoupled sGB field equations [Eq.~\eqref{eq:EOM-Decoupled}] in the background of a collapsing NS, which is obtained by solving the Einstein equations and the the conservation of energy-momentum equation in the matter sector. All of this is done through numerical relativity techniques. That is, we foliate the four-dimensional manifold $(M,g)$ into three-dimensional spatial hypersurfaces $(\Sigma_{t},\gamma)$, labelled by a time parameter $t$ and with an induced metric $\gamma_{ij}$.
By applying this spacetime decomposition to the field equations [Eq.~\eqref{eq:EOM-Decoupled}], we obtain hyperbolic (time-evolution) equations and elliptic equations, the latter of which we solve for the initial configurations.

The collapsing NS background is constructed as follows. First, we prescribe initial data as a (spherically-symmetric) solution to the Einstein equations. Assuming stationarity and spherical symmetry, they
reduce to the Tolman-Oppenheimer-Volkoff (TOV) equations \cite{1939PhRv...55..374O}
\begin{equation}
    \frac{{\rm d}P}{{\rm d} r} = - \left( e + P \right) \frac{m + 4 \pi r^3 P}{\left(1 - \frac{2m}{r}\right)}\,,
\end{equation}
\begin{equation}
    \frac{{\rm d}m}{{\rm d} r} = 4 \pi r^2 e\,,
\end{equation}
for which the metric is given by
\begin{equation}
    {\rm d}s^2 = - e^{2 \alpha} {\rm d} t^2 + \frac{1}{1 - \frac{2 m }{r}}  {\rm d} r ^2 + r^2 {\rm d} \Omega^2\,,
\end{equation}
\begin{equation}
    \frac{{\rm d}\alpha}{{\rm d} r} = \frac{m + 4\pi r^3 P}{r^2 \left(1 - \frac{2m}{r}\right)}\,.
\end{equation}
Here $e$ and $P$ are the energy density and pressure inside the NS, respectively. The mass $m$, denotes the mass contained within a radius $r$ from the center of the star.
In our numerical simulations we  conformally decompose the spatial part of the metric such that
\begin{equation}
    {\rm d}s^2 = - e^{2 \alpha} {\rm d} t^2 + \psi^4  \left[{\rm d} r ^2 + r^2 {\rm d} \Omega^2\right]\,,
\end{equation}
where $\psi^4$ is the conformal factor.
The NS matter sector is modeled through the H4 equation of state 
\cite{PhysRevLett.67.2414}, which is constructed from relativistic mean field theory, it includes hyperons in addition to neutrons, protons, electrons and muons, and it makes specific choices for the values of the incompressibility, the effective mass and the nucleon-meson coupling. We initialize the central density so that the solution to the TOV equations lead to a NS with mass $M = 2.02\, M_\odot$.
This data is unstable to gravitational collapse because the density inside the NS exceeds that of the maximum mass configurations with an H4 equation of state ($M_{\rm TOV, max H4} = 2.03\, M_\odot$) \footnote{We point out that the dynamically unstable branch of equilibrium NS solutions will have masses lower than the maximum mass, but reaches energy densities higher than dynamically stable NS.}
We do not add any additional perturbations, and instead, the collapse is triggered solely by the initial interpolation error from interpolating the spherical NS data onto the Cartesian grid.

Given this initial data, we then numerically solve the equations of general-relativistic magnetohydrodynamics \cite{Duez:2005sf} (in the absence of magnetic fields) and the Einstein equations to build the dynamical NS spacetime. 
We employ the Z4c formulation \citep{2010PhRvD..81h4003B,2013PhRvD..88h4057H}
together with the moving puncture gauge
using the \texttt{Frankfurt/IllinoisGRMHD} (\texttt{FIL}) code \cite{Most:2019kfe,Etienne:2015cea}.
\texttt{FIL} has been used extensively to study
BNS~\cite{Most:2018eaw,Most:2019onn,Most:2019pac,Papenfort:2021hod,Most:2021ktk} and NS-BH mergers \cite{Most:2019kfe,Most:2021ytn,Most:2020exl}, including effects of finite-temperatures, weak interactions and magnetic fields.
\texttt{FIL} implements highly-accurate numerical methods \cite{DelZanna:2007pk}
and has been demonstrated to achieve beyond second-order convergence of the numerical solution in simulations of NS space-times \cite{Most:2019kfe} .

This
solution provides a numerically-generated spacetime background on which we can solve the scalar field evolution equation in sGB gravity [Eq.~\eqref{eq:sGB-OS}].
We evolve Eq.~\eqref{eq:sGB-OS} as a time evolution problem 
using \canuda~\cite{witek_helvi_2021_5520862}, an open-source software for numerical relativity beyond GR~\cite{Benkel-2016,Benkel_2017,Witek:2018dmd,Silva:2020omi} and to simulate BHs interacting with ultra-light fields~\cite{Okawa:2014nda,Zilhao:2015tya};
see Refs.~\cite{Benkel-2016,Witek:2018dmd} for details on the implementation of the sGB equations.
{\canuda} is compatible with the {\texttt{Einstein Toolkit}} software for computational astrophysics
\cite{steven_r_brandt_2021_5770803,Loffler:2011ay,Zilhao:2013hia}.

The numerical domain is constructed as a series of $5$ nested boxes with an outer domain size of $128\, M$ and a finest resolution of $0.05\, M$. Note, that we work in geometric units, $c=1=G$ and measure lengths and evolution time of the numerical simulation in units of the mass $M$, with $M=2.02M_{\odot}$. 
We initialize the scalar field to zero, but note that given the small initial perturbation, 
the star does not immediately collapse, so that 
the numerical solution of the scalar field settles to a stationary configuration for a TOV space-time before the NS begins to contract.

During the numerical evolution and during post-processing we compute a number of quantities that aid us in understanding the behaviour of the scalar hair.
One of them is the apparent horizon, which we calculate using the \texttt{AHFinderDirect} code \cite{Thornburg:2003sf} on every time slice. Another one is the event horizon, which we calculate using the approach described in~\cite{Diener:2003jc}, by tracing geodesics back in time post evolution of the NS background. The last quantity we need is a spacetime diagram of the scalar field and associated quantities. We construct this diagram by focusing only on one of the coordinate axis, because, although our simulations are performed in three spatial dimensions, the spherical symmetry of all solutions is well preserved during collapse.

The resulting evolution is shown in Fig. \ref{fig:numerical_collapse} where we present the relative scalar field amplitude (left panel) and energy flux (right panel). 
We mark the surface of the star as a contour enclosing $99\%$ of the baryon mass,
\begin{equation}
    M_b = \int {\rm d}V \sqrt{\gamma} \rho_B \Gamma\,,
\end{equation}
where $\rho_B$ is the baryon density and $\Gamma$ the Lorentz factor of the NS material.
After an initial transient (not shown), we can see that a stationary scalar field solution inside the star develops. As expected from Eq. \eqref{eq:OS_final_sol}, 
the exterior of the star does not contain a $r^{-1}$ contribution at initial times (black area).

Around $t \approx 70\, M$ the event horizon is first 
found at the center of the star, and gradually expands outwards. An apparent horizon is not found until all the matter has crossed the event horizon. This happens around $t \approx 90\, M$. We can see that around the same time the scalar hair begins to grow (yellow region). The singularity does not begin to form until $t \approx 100\, M$. This numerically confirms that the growth of scalar hair is not associated with the formation of a singularity, or with the formation of the apparent horizon, but rather with essentially the NS surface disappearing behind the event horizon.

In order to better understand at which point the scalar field begins to grow, we also calculate the amount of radial energy flux contained in the scalar field. We do this by computing the analogue of the Poynting flux in electromagnetism, namely
\begin{align}
    S^r_\Phi = \partial^t \Phi \partial_r \Phi + \frac{1}{2}g^t_r \partial_\mu \Phi \partial^\mu \Phi.
\end{align}
Since a change in the scalar field will
yield an associated
energy flux, we can use this quantity to track the evolution of the scalar field. The resulting evolution is shown in the right panel of Fig. \ref{fig:numerical_collapse}. Since the dynamical coordinate evolution in the simulation leads to a tilting of light cones close to the center of the star, we also show an out-going null geodesic (light blue curve), in order to show when changes inside the star propagate to the surface. We can see that the initial change in the scalar field, measured in terms of an out-going energy flux, is caused by the contraction in the center of the star. However, the main growth phase of the scalar field does not happen until the event horizon has formed. This can most easily be seen by the maximum (out-going) scalar flux on the horizon, as marked by a black dot in Fig. \ref{fig:numerical_collapse}.  We emphasize, that although we are numerically evolving the interior of the BH, everything within the event horizon (blue curve) is causally disconnected from the exterior solution.
However, we can also see that this initial flux does not reach the surface until some time later, around $t \approx 80\, M$.
At this time, the star has not yet completely collapsed to a BH, and the $r^{-1}$ term of the scalar field has not yet grown, as is revealed by a comparison with the left panel. In terms of the energy flux, no apparent transition can be seen later on. The peak of the out-going energy flux (black dot) is reached only when the entire star is about to be swallowed.  
\section{Conclusions and Future Directions}\label{sec:Conclusions}

We have investigated the dynamics of scalar fields during spherically-symmetric gravitational collapse in a wide class of modified gravity theories.
We proved that BH solutions in these theories generically fall into two different classes, depending on whether the scalar field they support is regular at the EH of the BH or not.
We also proved that if the scalar field is regular at the EH, then the surface gravity of the EH is non-zero.
Physically, one expects that the scalar field will settle to a regular solution during gravitational collapse.
This flow to regularity was investigated numerically in sGB gravity by simulating the sGB scalar in the background of the OS collapse~\cite{Benkel-2016,Benkel_2017}.
We confirm and greatly extend these results by proving that, during spherically symmetric collapse, the scalar field eventually settles to a regular solution to first order in perturbation theory and in a wide class of modified gravity theories where a scalar couples to a curvature invariant.
As a corollary, we also show how this flow to regularity affects the scalar hair in the initial NS state and the final BH state.

To understand the physical reasons behind these results, we investigated the dynamics of the scalar field evolution in two background spacetimes and in sGB theory as an example.
In this theory, the scalar hair goes from zero for a NS to a non-zero value for a BH during gravitational collapse.
We verified that this happens by finding an exact analytical solution for the scalar field in the interior of the star during OS collapse.
Using this solution, we showed that the regularity of wave modes before the EH forms ensures that the scalar field remains regular during collapse. After the star collapses to a BH, the wave modes decay away by Price's law~\cite{Gundlach:1993tp,Price-1972}, leading to a regular static solution to the sGB equations that is hairy. We confirm these results in the more realistic scenario by simulating the collapse of a non-rotating NS in full general-relativity. On top of this dynamical background we solved the decoupled sGB equations to study the evolution of the scalar field. In line with our analytical results, we find that escaping modes of the scalar field form inside the star but outside of the growing EH. We further confirm that the formation of scalar hair roughly coincides with the bulk of the matter disappearing behind the EH, generalizing the results from OS to realistic NS models.

Our work points to a few natural directions to pursue in the future.
One natural direction to explore is the dynamics of the scalar field in axisymmetric collapse and thereby extending studies of the scalar field's formation around a Kerr background~\cite{Witek:2018dmd}.
The results we established in spherically-symmetric collapse should be extendable to the axisymmetric case, and this will help us analyze theories such as dCS gravity which reduce identically to GR in spherical symmetry.
A second direction for future work is to extend our results to higher orders in perturbation theory by including the back reaction of the scalar field onto the metric. 
This would be complementary to (and extending) previous simulations of scalar field collapse in the full theory~\cite{Ripley-Pretorius-2019} and studies of the metric perturbations at second order in perturbation theory~\cite{Okounkova:2019zep,Okounkova:2020rqw}.
Since the class of theories we consider in Eq.~(\ref{eq:Action}) reduce to GR when the coupling constant goes to zero, the results we have established here should go through, with small second-order modifications. 
Another direction for future work includes studying the behaviour of dipole and higher order moments of the scalar field during gravitational collapse.
A final natural direction to pursue is to investigate whether our results can be extended to a wider class of theories that exhibit other non-linear deviations from GR in the strong coupling regime.
Examples of these include certain scalar tensor theories~\cite{1996-Damour} and theories where the coupling function scales at least as $\Phi^2$ instead of the linear coupling we considered in Eq.~(\ref{eq:Action}).
These theories are known to exhibit phenomena such as spontaneous scalarization,
which are intrinsically non-perturbative~\cite{1996-Damour,2018-Hector-Spontaneous,2018-Daniela}.
Understanding the behavior of scalar hair in these theories will help us determine the best routes to constrain them.

\begin{acknowledgements}
The authors thank Professors Robert Wald, Piotr Chru\'sciel and Frans Pretorius for helpful discussions and comments. AH and NY acknowledge support from the Simons Foundation through Award number 896696.
ERM acknowledges support from postdoctoral fellowships at the Princeton Center for Theoretical Science, the Princeton Gravity Initiative, and the Institute for Advanced Study.
JN is partially supported by the U.S.
Department of Energy, Office of Science, Office for Nuclear Physics under Award No. DE-SC0021301.
HW acknowledges financial support provided by NSF grants No. OAC-2004879 and No. PHY-2110416, and Royal Society (UK) Research Grant RGF\textbackslash R1\textbackslash180073.

This work used the Blue Waters sustained-petascale computing project
which is supported by National Science Foundation awards No. OCI-0725070
and No. ACI-1238993, the State of Illinois and the National Geospatial Intelligence Agency.
Blue Waters is a joint effort of the University of Illinois at Urbana-Champaign and its National Center for Supercomputing Applications.
This work used the Extreme Science and Engineering Discovery Environment (XSEDE), which is supported by National Science Foundation grant number ACI-1548562. The authors acknowledge the Texas Advanced Computing Center (TACC) at The University of Texas at Austin for providing HPC resources that have contributed to the research results reported within this paper, under LRAC grants AT21006.
Part of the simulations presented in this article were performed on computational resources managed and supported by Princeton Research Computing, a consortium of groups including the Princeton Institute for Computational Science and Engineering (PICSciE) and the Office of Information Technology's High Performance Computing Center and Visualization Laboratory at Princeton University.
\end{acknowledgements}

\input{Appendix}
\bibliography{ref}
\end{document}

%% file: Appendix.tex
\appendix
\section{Proof of Connection Between Non-Vanishing of Surface Gravity and Regularity of the Scalar Field} \label{Appendix:Proof}
In this Appendix, we prove that if the scalar field is regular on the EH then surface gravity of the EH must be non-zero. 
As noted in Sec.~\ref{sec:Case-1-Scalar-Field}, if $\mathcal{F}(0) \neq 0$, then regularity at the EH implies that $p=0$ and the surface gravity is finite. 
If $\mathcal{F}(0) = 0$, then we have to look at higher order terms. 
We show that looking at higher order terms leads to a contraction if $p\geq 1$, i.e., if the surface gravity vanishes.
We prove that this is true for two general choices for the curvature scalar $\mathcal{F}$.
For the first case, we consider modified quadratic theories of gravity and in the second case we consider general $\mathcal{F}(R)$ theories.
The idea behind the proof uses the fact that if the cross section of EH is a 2-sphere then its Gaussian curvature given by $K(0)^{-1}$ cannot be zero. This idea is inspired by Ref.~\cite{Chrusciel-Reall-Tod-2005}, where a similar argument was used to show that there can be no static BH solutions in GR with a degenerate EH.

\subsection{Modified Quadratic Gravity}
In case of modified quadratic theories of gravity, the curvature scalar $\mathcal{F}$ takes the form,
\begin{equation}\label{eq:Modified-Grav-F}
    \mathcal{F} = \alpha_1 R^2 + \alpha_2 R_{\mu\nu}^2 + \alpha_3 R_{\alpha\beta\gamma\delta}^2,
\end{equation}
the field equations for this theory are obtained by replacing,  
\begin{align}
    \alpha_i \to \epsilon\, \alpha_i ,\\
    \beta = 1,\\
    \alpha_4 = 0.
\end{align}
in Eqs.\ (27-28) in Ref.~\cite{Yunes_2013}.
We now have two cases to deal with:

\vspace{5mm}
\textbf{Case 1: $\mathcal{F}(0) = 0$, $\mathcal{F}'(0) \neq 0$}

\vspace{2mm}
In this case, the local expansion of Eq.(\ref{eq:Scalar-Field-Dr}) in Gaussian null coordinates takes the form,
\begin{equation}
    \partial_r \Phi = -\frac{\epsilon\,\mathcal{F}'(0)}{2D_0(0)}r^{1-p} + O(r^{2-p}),
\end{equation}
regularity requires that, $p \in \{0,1\}$ since we assume that $\mathcal{F}'(0) \neq 0$. 
If $p = 0$, then surface gravity is finite and we are done. Now suppose $p = 1$. Then we need to use the gravitational equations.
Let us denote the gravitational equations by,
\begin{align}
    E_{\mu\nu} &:=  G_{\mu\nu} + 16\pi \,\epsilon\, \mathcal{C}_{{\mu\nu}} - 8\pi T_{\mu\nu}^{\Phi} =0, \\
    E &:= g^{\mu\nu}E_{\mu\nu} = 0.
\end{align}
Taking the trace of $E_{\mu\nu}$ and evaluating it at $r=0$, we get,
\begin{equation}\label{eq:Trace-loc-Mod-Quad-Grav}
    E(0) = 2 D_0(0) - \frac{2}{K(0)} = 0 \implies D_0(0) = \frac{1}{K(0)}.
\end{equation}
We can now look at the value of $E_{vr}$ at $r=0$ to get a contradiction.
\begin{dmath}
   E_{vr} = -\frac{16 \pi\,\epsilon\,  \Phi (0)  [2 \alpha_1+\alpha_2+2 \alpha_3]}{K(0)^2}-\frac{1}{K(0)}+16 \,\pi \,\epsilon\,  D_0(0){}^2\,\Phi (0) \,  [2 \alpha_1+\alpha_2+2 \alpha_3] = 0,
\end{dmath}
from Eq.\ (\ref{eq:Trace-loc-Mod-Quad-Grav}) we know that $D_0(0) = K(0)^{-1}$. Using this in the above equation we get,
\begin{equation}
    \frac{1}{K(0)} = 0,
\end{equation}
this is a contradiction because $K(0)^{-1}$ is the Gaussian curvature of the 2-sphere at the horizon, which cannot be zero by assumption (3).

\vspace{5mm}
\textbf{Case 2: $\mathcal{F}^{(m)}(0) =0 , \mathcal{F}^{(m+1)}(0) \neq 0$}

\vspace{2mm}
Using the definition
\begin{equation}
    \mathcal{F}^{(m+1)}(0) := \left. \frac{\partial^{m+1} \mathcal{F}}{\partial r^{m+1}} \right|_{r=0},
\end{equation}
the local expansion of Eq.\ (\ref{eq:Scalar-Field-Dr}) takes the form,
\begin{equation}
    \partial_r \Phi = -\frac{\epsilon\,\mathcal{F}^{(m+1)}(0)}{2\,D_0(0)}r^{m+1-p} + O(r^{m+2-p}).
\end{equation}
Therefore, regularity requires that $p\in \{0,1,..,m+1\}$ from Eq.\ (\ref{eq:Surface-Gravity}). 
If $p=0$ then we see that surface gravity is nonzero. If $p=1$, we already showed a contradiction in Case 1 above. 
Now suppose $p \geq 2$. The value of $E$ at the EH now takes the form,
\begin{equation}
   E(0) = -\frac{2}{K(0)} = 0,
\end{equation}
which gives us a contradiction as before.
\subsection{$\mathcal{F}(R)$ Theories}
For theories where the curvature coupling $\mathcal{F}$ is only a function of the Ricci scalar $R$, the $C_{\mu\nu}$ tensor takes the form,
\begin{dmath}
    C_{\mu\nu} = \mathcal{F}_1\, \Phi \,R_{\mu\nu} + g_{\mu\nu}\, \Box(\mathcal{F}_1\, \Phi) - \nabla_{\mu}\nabla_{\nu}(\mathcal{F}_1\, \Phi) -\frac{g_{\mu\nu}}{2} \,\mathcal{F}\, \Phi,
\end{dmath}
where
\begin{equation}
    \mathcal{F}_1 := \frac{\partial \mathcal{F}}{\partial R}.
\end{equation}
We now assume that $\mathcal{F}$ can be represented as a power series in $R$,
\begin{equation}
    \mathcal{F} = \sum_{n=0}^{N}{\alpha_n} R^n, \quad N>1.
\end{equation}
We also make another regularity assumption.
We assume that the scalar field and the metric coefficients are regular functions of the coupling constant $\epsilon$ at the EH. Note that this assumption \textit{does not} mean that we are working in a small coupling expansion. 
It means that the theory we are considering reduces to GR when the coupling constant goes to zero.
We now have two cases to analyse as before.

\vspace{5mm}
\textbf{Case 1: $\mathcal{F}(0) = 0$, $\partial_r \mathcal{F}(0) \neq 0$}

\vspace{2mm}
In this case the local expansion of Eq.\ (\ref{eq:Scalar-Field-Dr}) takes the form,
\begin{equation}
    \partial_r \Phi = -\frac{\epsilon\,\mathcal{F}'(0)}{2D_0(0)}r^{1-p} + O(r^{2-p}).
\end{equation}
Regularity requires that $p \in \{0,1\}$ since we assume that $\mathcal{F}'(0) \neq 0$. 
If $p = 0$ then surface gravity is non-zero and we are done.
Suppose $p = 1$. Then, we need to use the gravitational equations.
Let us look at the value of $E$ at the EH,
\begin{equation}
  E(0)  =-\frac{2 \left[16\, \pi\, \epsilon\, \mathcal{F}_1(0)\, \Phi (0)  -1\right] \left[D_0(0) K(0)-1\right]}{K(0)}.
\end{equation}
Because of the regularity assumption, 
\begin{equation}\label{eq:Local-Ricci-Flat-F-R}
    16 \pi \epsilon \mathcal{F}_1(0) \Phi (0)  \neq 1 \implies D_0(0) = \frac{1}{K(0)}.
\end{equation}
The local expansion of $E_{vr}$ is gives us,
\begin{equation}
   E_{vr}(0) =  -\frac{1}{K(0)}-8 \pi \,\epsilon\, \Phi (0)  \,\left[2 D_0(0) F_1(0)+F(0)\right] 
\end{equation}
We now use the fact that $F(0)=0$ and $D_0(0) = K(0)^{-1}$ in the above equation to get,
\begin{align}
    \frac{16 \pi\,\epsilon  \mathcal{F}_1(0)\, \Phi (0)  +1}{K(0)} = 0, \nonumber\\
    16 \pi \,\epsilon \mathcal{F}_1(0)\, \Phi (0)  +1 \neq 0, \nonumber\\
    \implies \frac{1}{K(0)} = 0 \nonumber,
\end{align}
which is a contradiction.

\vspace{5mm}
\textbf{Case 2: $\mathcal{F}^{(m)}(0) =0 , \mathcal{F}^{(m+1)}(0) \neq 0$}

\vspace{2mm}
In this case the local expansion of Eq.\ (\ref{eq:Scalar-Field-Dr}) takes the form,
\begin{equation}
    \partial_r \Phi = -\frac{\epsilon\,\mathcal{F}^{(m+1)}(0)}{2D_0(0)}r^{m+1-p} + O(r^{m+2-p}).
\end{equation}
Therefore, $p\in \{0,1,..,m+1\}$ if the field is regular. 
If $p=0$, then we see that surface gravity is nonzero from Eq.(\ref{eq:Surface-Gravity}). If $p=1$ we already showed a contradiction in Case 1 above. Suppose now that $p \geq 2$. The local expansion of $E_{vr}$ now takes the form,
\begin{equation}
   E_{vr} =  -\frac{1}{K(0)}-8 \pi \,\epsilon \,\mathcal{F}(0) \Phi (0)  = 0,
\end{equation}
using $\mathcal{F}(0) = 0$, we end up with a contradiction as before.
\par We have checked that the above derivation is valid when the curvature scalar takes the form,
\begin{equation}
     \mathcal{F} = \alpha_1 \,R + \alpha_2 \,R_{\mu\nu}^2 + \alpha_3\, R_{\alpha\beta\gamma\delta}^2.
\end{equation}
\section{Coordinate transformations between PG coordinates and Schwarzschild and FLRW coordinates}\label{Appendix:Coordinate-Transformations}
In this Appendix we give the coordinate transformations between global PG coordinates introduced in Sec.~\ref{sec:OS-Spacetime} and FLRW coordinates inside the star and Schwarzschild spacetime outside the star.
\subsection{Interior of the Star and FLRW Coordinates}
We denote the FLRW coordinates by $x^{\mu} = \left(t,r_c,\theta,\phi\right)$. 
The metric inside the star in these coordinates takes the form,
\begin{equation}
    ds^2 = -dt^2 + a^2(t)(dr_c^2 + r_c^2 d\Omega^2),
\end{equation}
where the scaling factor $a(t)$ is given by,
\begin{equation}
    a(t) = \left(\frac{-t}{M}\right)^{2/3}.
\end{equation}
The coordinate transformation to PG coordinates takes the form,
\begin{align}\label{eq:PG-FLRW}
    t = t, \,
    r(t,r_c) = r_c a(t), \,
    \theta = \theta , \,
    \phi = \phi.
\end{align}
\subsection{Exterior of the Star and Schwarzschild Coordinates}
The spacetime outside the star is the Schwarzschild spacetime by Birkhoff's theorem.
The transformation from Schwarzschild coordinates $x^{\mu} = (t_s,r,\theta,\phi)$ to PG coordinates is given by,
\begin{align}
    &t(t_s,r) = t_s + 4M\left[\sqrt{\frac{r}{2M}} -\frac{1}{2}\log\left(\frac{\sqrt{r}+\sqrt{2M}}{\sqrt{r}-\sqrt{2M}}\right)\right] \nonumber , \\
    &r = r , \, \theta = \theta, \, \phi = \phi .
\end{align}